\documentclass[12pt,draftclsnofoot,onecolumn]{IEEEtran}

\usepackage{epsfig}
\usepackage{amsthm,amsfonts}
\usepackage{amsmath,amssymb,setspace,cite,color,graphicx}
\usepackage[caption=false,font=footnotesize]{subfig} 

\allowdisplaybreaks[1]

\newtheorem{remark}{Remark}
\newtheorem{proposition}{Proposition}
\newtheorem{theorem}{Theorem}

\newtheorem{corollary}{Corollary}
\newtheorem{definition}{Definition}

\begin{document}


\title{\onehalfspacing{Lossy Coding of Correlated Sources over a Multiple Access Channel: Necessary Conditions and Separation Results}}

\author{\IEEEauthorblockN{Ba\c{s}ak~G{\"u}ler\IEEEauthorrefmark{1}  \quad Deniz~G{\"u}nd{\"u}z \IEEEauthorrefmark{7}    \quad Aylin~Yener\IEEEauthorrefmark{1}} 
\vspace{-0.3cm}\singlespacing{\IEEEauthorblockA{\IEEEauthorrefmark{1} 
Department of Electrical Engineering\\
The Pennsylvania State University\\
University Park, PA \\
{\em basak@psu.edu \qquad yener@ee.psu.edu}} 
}
\vspace{-0.4cm}\singlespacing{\IEEEauthorblockA{\IEEEauthorrefmark{7}Department of Electrical and Electronic Engineering \\ Imperial College London \\
London, UK \\
{\em d.gunduz@imperial.ac.uk}} \\
}
\thanks{
The material in this paper was presented in part at the 2016 IEEE International Symposium on Information Theory (ISIT'16) and the 2017 IEEE International Symposium on Information Theory (ISIT'17).} 
\thanks{This research is sponsored in part by the U.S. Army Research Laboratory under the Network Science Collaborative Technology Alliance, Agreement Number
W911NF-09-2-0053, and by the European Research Council Starting Grant project BEACON (project number 677854).}
}

\maketitle

\vspace{-1.6cm} 

\begin{abstract}  
\vspace{-0.2cm}Lossy coding of correlated sources over a multiple access channel (MAC) is studied. First,  a joint source-channel coding scheme is presented when the decoder has correlated side information. Next, the optimality of separate source and channel coding, that emerges from the availability of a common observation at the encoders, or side information at the encoders and the decoder, is investigated. It is shown that separation is optimal when the encoders have access to a common observation whose lossless recovery is required at the decoder,  and the two sources are independent conditioned on this common observation. Optimality of separation is also proved when the encoder and the decoder have access to shared side information conditioned on which the two sources are independent. These separation results obtained in the presence of side information are then utilized to provide a set of necessary conditions for the transmission of correlated sources over a MAC without side information. 
Finally, by specializing the obtained necessary conditions to the transmission of binary and Gaussian sources over a MAC, it is shown that they can potentially be tighter than the existing results in the literature, providing a novel converse for this fundamental problem.

\end{abstract}

\section{Introduction}
\label{Sec:introduction} 
This paper considers the lossy coding of correlated discrete memoryless (DM) sources over a DM multiple access channel (MAC). 
Separate source and channel coding is known to be suboptimal for this setup in general, even when the lossless reconstruction of the sources is required \cite{cover1980multiple}. This is in contrast to the point-to-point scenario for which the separation of source and channel coding is optimal, also known as the \emph{separation theorem} \cite{Shannon}. The characterization of the achievable distortion region when transmitting correlated sources over a MAC is one of the fundamental open problems in network information theory, solved only for some special cases. 

This problem is also related to another long-standing open problem, namely the multi-terminal lossy source-coding problem, which refers to the scenario when the underlying MAC consists of two orthogonal finite-capacity error-free links. 
Despite the lack of a general single-letter characterization for the multi-terminal source coding problem, separate source and channel coding is  optimal when the underlying MAC is orthogonal \cite{xiao2007multiterminal}. Separation is also optimal when one of the sources is shared between the two encoders \cite{gunduz2007correlated}, or for the lossless case, when the decoder has access to side information conditioned on which the two sources are independent \cite{gunduz2009source}. However, due to the lack of a general separation result, the achievable distortion region is unknown even in scenarios for which the corresponding source coding problem can be solved. 

In the absence of single-letter necessary and sufficient conditions, the goal is to obtain computable inner and outer bounds. A fairly general joint source-channel coding scheme was introduced in \cite{minero2015unified} by leveraging hybrid coding. This scheme subsumes most other known coding schemes. A novel outer bound was presented in \cite{lapidoth2010sending} for the Gaussian setting, which uses the fact that the correlation among channel inputs is limited by the correlation available among source sequences. Other bounds were proposed in \cite{jain2012energy}, \cite{kang2011new}, and more recently in \cite{7541654}, \cite{yu2016distortion}. 
Optimality of source-channel separation was studied in \cite{gunduz2009source}, \cite{tian2014optimality}, and the optimality of uncoded transmission was investigated for Gaussian sources over multi-terminal Gaussian channels in \cite{tian2017matched}.

This paper studies the achievable distortion region for sending correlated sources over a MAC. In the first part of the paper, it is assumed that the encoders and/or the decoder may have access to side information correlated with the sources (see Fig.~\ref{Fig:Model1}). Initially, a joint source-channel coding scheme is proposed when side information is available only at the decoder. Then, we investigate separation theorems that emerge from the availability of a common observation at the encoders, or from the availability of side information at the encoders and the decoder. 
In doing so, we first focus on the scenario in which the encoders share a common observation conditioned on which the two sources are independent. For this setup, we show that separation is optimal when the decoder is required to recover the common observation losslessly, but can tolerate some distortion for the parts known only at a single encoder. 
Corresponding necessary and sufficient conditions are identified for the optimality of separation. 
Next, we consider the scenario in which the encoders and the decoder have access to shared side information, and show that separation is again optimal if the two sources are conditionally independent given the side information. 

In the second part of the paper, we leverage the separation theorems derived in the first part to obtain a new set of necessary conditions for the achievability of a distortion pair when transmitting correlated sources over a MAC without any side information. In particular, we obtain our computable necessary conditions by providing particular side information sequences to the encoders and the decoder to induce the optimality of separation. Based on the results of the first part, this can be achieved when the two sources are conditionally independent given the side information. Optimality of separation conditioned on the provided side information allows us to characterize the corresponding necessary conditions explicitly. Conditional independence inducing side information sequences have previously been used to obtain converse results in some multi-terminal source coding problems \cite{ozarow80, wagnerIT08}. 
In this paper, they are used to obtain converse results in a multi-terminal joint source-channel coding problem. The necessary conditions are then specialized to the case of bivariate Gaussian sources over a Gaussian MAC as well as doubly symmetric binary sources (DSBS) over a Gaussian MAC. By providing comparisons between the new necessary conditions and the known bounds in the literature, we show that the proposed technique can potentially provide tighter converse bounds than the previous results in the literature.

In the remainder of the paper,  $X$ represents a random variable, and
$x$ is its realization. $X^n=(X_1, \ldots, X_n)$ is a random vector of length $n$, and $x^n=(x_1, \ldots, x_n)$ denotes its realization. $\mathcal{X}$ is a set with 
cardinality $|\mathcal{X}|$. $\mathbb{E}[X]$ is the expected value  and $\text{var}(X)$ is the variance of $X$.

\begin{figure}[t]
\centering
\includegraphics[width=0.7\linewidth]{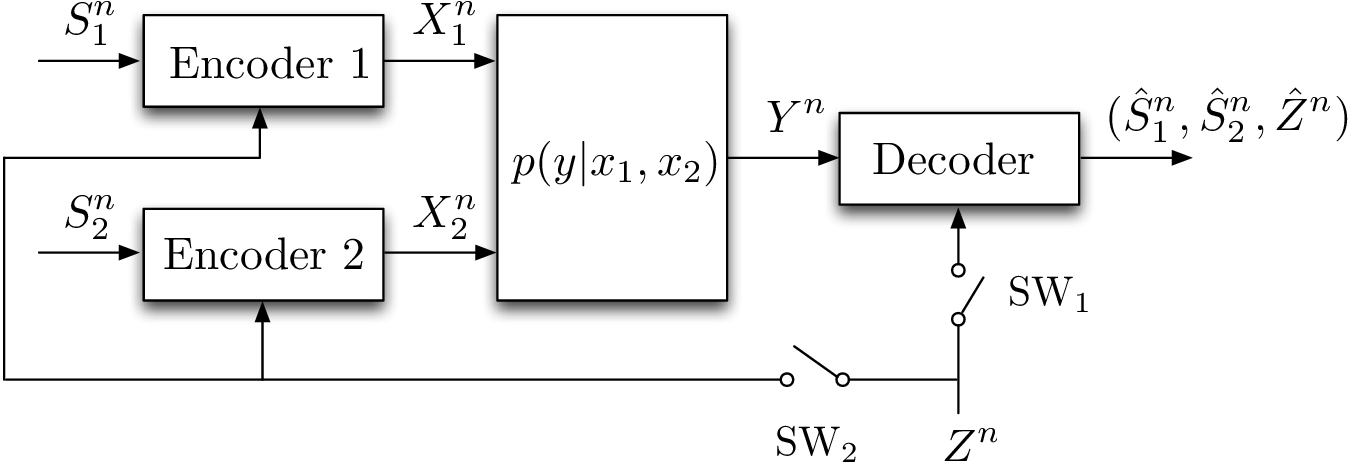}
\vspace{-0.2cm}\caption{Communication of correlated sources over a MAC.}
\label{Fig:Model1}    
\vspace{-0.5cm}\end{figure}

\section{System Model}\label{Sec:SystemModel}

We consider the transmission of DM sources $S_1$ and $S_2$ over a DM MAC as illustrated in Fig.~\!\!\ref{Fig:Model1}. 
Encoder $1$ observes $S_1^n=(S_{11}, \ldots, S_{1n})$, whereas encoder $2$ observes $S_2^n=(S_{21}, \ldots, S_{2n})$. If switch $\text{SW}_2$ in Fig.~\!\ref{Fig:Model1} is closed, the two encoders also have access to a common observation $Z^n$ correlated with $S_1^n$ and $S_2^n$. 
Encoders $1$ and $2$ map their observations to the channel inputs $X_1^n$ and $X_2^n$, respectively. The channel is  characterized by the conditional distribution $p(y|x_1, x_2)$. If switch $\text{SW}_1$ in Fig.~\!\ref{Fig:Model1} is closed, the decoder has access to side information $Z^n$.  
Upon observing the channel output $Y^n$ and side information $Z^n$ whenever it is available, the decoder constructs the estimates $\hat{S}_1^n$, $\hat{S}_2^n$, and $\hat{Z}^n$. Corresponding average distortion values for the source sequence $\hat{S}_j^n$, $j=1,2$, is given by 
\begin{equation}
\Delta^{(n)}_j = \frac{1}{n}\sum_{i=1}^n\mathbb{E}[d_j(S_{ji}, \hat{S}_{ji})], 
\end{equation}
where $d_j(\cdot, \cdot)< \infty$ is the distortion measure for source $S_j^n$.  
A distortion pair $(D_1, D_2)$ is \emph{achievable} for the source pair $(S_1, S_2)$ and channel $p(y|x_1, x_2)$ if there exists a sequence of encoding and decoding functions such that 
\vspace{-0.15cm}\begin{equation}
\limsup_{n\rightarrow \infty} \Delta^{(n)}_j \leq D_j, \quad j=1,2,
\vspace{-0.1cm}\end{equation}
and $P(Z^n\neq \hat{Z}^n)\rightarrow 0$ as $n\rightarrow \infty$ when at least one of the switches is closed. 
Random variables $S_1$, $S_2$, $Z$, $X_1$, $X_2$, $Y$, $\hat{S}_1$, $\hat{S}_2$, $\hat{Z}$  
are defined over the corresponding alphabets $\mathcal{S}_1$, $\mathcal{S}_2$, $\mathcal{Z}$, $\mathcal{X}_1$, $\mathcal{X}_2$, $\mathcal{Y}$, $\hat{\mathcal{S}}_1$, $\hat{\mathcal{S}}_2$, $\hat{\mathcal{Z}}$. 
Note that, when switch $\text{SW}_1$ is closed, error probability in decoding $Z^n$ becomes irrelevant since it is readily available at the decoder, and serves as side information.

Throughout the paper, we use the following definitions extensively. 
\begin{definition}{(Conditional rate distortion function)\cite{gray1972conditional}}
Given correlated random variables $S$ and $U$, define the minimum average distortion for $S$ given $U$ as  \cite{gunduz2007correlated}, \cite{shamai1998systematic}:
\vspace{-0.1cm}\begin{equation} \label{func1}
\mathcal{E}(S|U)=\inf_{f:\mathcal{U}\rightarrow \hat{\mathcal{S}}} E[d(S, f(U))], 
\end{equation}
where the minimum is over all functions $f(\cdot)$ from $\mathcal{U}$ to the reconstruction alphabet $\hat{\mathcal{S}}$. Then, the conditional rate distortion function for source $S$ when correlated side information $Z$ is shared between the encoder and the decoder is given by,
\begin{equation}\label{func2}
R_{S|Z} (D) = \min_{\substack{p(u| s, z):\\\mathcal{E}(S|U, Z)\leq D}} I(S; U|Z), 
\end{equation}
where the minimum is over all conditional distributions $p(u|s,z)$ such that the minimum average distortion for $S$ given $U$ and $Z$ is less than or equal to $D$.
\end{definition}

\begin{definition}{(G{\'a}cs-K{\"o}rner common information)\cite{gacs1973common}}\label{GKW}
Define the function $f_j: \mathcal{S}_j\rightarrow \{1, \ldots, k\}$ for $j=1,2$, with the largest integer $k$ such that $P(f_j(S_j)=u_0)>0$ for $u_0\in \{1, \ldots, k\}$, $j=1,2$, and $P(f_1(S_1)=f_2(S_2)) = 1$. Then, $U_0=f_1(S_1)=f_2(S_2)$ is defined as the common part between $S_1$ and $S_2$, and the G{\'a}cs-K{\"o}rner common information is given by
\begin{equation}
C_{GK}(S_1, S_2)=H(U_0) \label{GK}.
\end{equation} 
\end{definition}

  \begin{definition}{(Wyner's common information)\cite{wyner1975common}}
 Wyner's common information between $S_1$ and $S_2$ is defined as,
\begin{equation} \label{commoninfo}
 C_{W}(S_1, S_2) = \min_{\substack{p(v|s_1, s_2) \\ S_1-V-S_2}} I(S_1, S_2; V).
\end{equation}
 \end{definition}

\section{Joint Source-Channel Coding with Decoder Side Information}\label{joint}
We first assume that only $\text{SW}_1$ is closed in Fig.~\!\ref{Fig:Model1}, and present a general achievable scheme for the lossy coding of correlated sources in the presence of decoder side information.

\begin{theorem}\label{lemma:hybrid}
When sending correlated DM sources $S_1$ and $S_2$ over a DM MAC with $p(y|x_1, x_2)$ and decoder side information $Z$, distortion pair $(D_1, D_2)$ is achievable if there exists a joint distribution $p(u_1, u_2, s_1, s_2, z)=p(u_1|s_1)p(u_2|s_2)p(s_1, s_2, z)$, and functions $x_j(u_j, s_j)$, $g_j(u_1, u_2, y,z)$ for $j=1,2$, such that
\begin{align}
I(U_1; S_1|U_2, Z)&< I(U_1;Y|U_2, Z) \label{eq4n}\\
I(U_2; S_2|U_1, Z)&< I(U_2;Y|U_1, Z) \label{eq5n}\\
I(U_1, U_2; S_1, S_2| Z)&<I(U_1, U_2;Y|Z) \label{eq6n}
\end{align}
and $\mathbb{E}[d_j(S_j, g_j(U_1, U_2, Y, Z))]\leq D_j$ for $j=1,2$. 
\end{theorem}

\begin{proof}
Our achievable scheme builds upon the hybrid coding framework of  \cite{minero2015unified}, by generalizing it to the case with decoder side information. The detailed proof is available in Appendix~\ref{appendix0}. 
\end{proof}

\section{Separation Theorems}\label{section3} 
We now focus on the conditions under which separation is optimal for lossy coding of correlated sources over a MAC. 
For the remainder of this section, we assume that $S_1$ and $S_2$ are independent conditioned on $Z$, i.e., the Markov condition $S_1-Z-S_2$ holds.

\subsubsection{Separation in the Presence of Common Observation}

Here, we assume that only switch $\text{SW}_2$ in Fig.~\!\ref{Fig:Model1} is closed, and show the optimality of separation if the lossless reconstruction of the common observation $Z$ is required.

\begin{theorem} \label{lemma3}
Consider the communication of correlated sources $S_1$, $S_2$, and $Z$,  where $Z$ is observed by both encoders. If $S_1-Z-S_2$ holds, then separation is optimal, and $(D_1, D_2)$ is achievable if 
\begin{align}
R_{S_1|Z} (D_1) &< I(X_1;Y|X_2, W) \label{common1}\\
R_{S_2|Z} (D_2) &< I(X_2;Y|X_1, W)\label{common2}\\
R_{S_1|Z} (D_1)+R_{S_2|Z} (D_2) &< I(X_1, X_2;Y| W)\label{common3}\\
H(Z) +R_{S_1|Z} (D_1)+R_{S_2|Z} (D_2) &< I(X_1, X_2;Y) \label{common4}
\end{align}
for some $p(x_1, x_2, y, w)=p(y|x_1, x_2) p(x_1|w) p(x_2|w)p(w)$.

Conversely, if a distortion pair $(D_1, D_2)$ is achievable, then \eqref{common1}-\eqref{common4} must hold with $<$ replaced with $\leq$.
\end{theorem}

\begin{proof} 
We provide a detailed proof in Appendix~\ref{appendixB}. 
\end{proof}

\begin{corollary}
A special case of Theorem~\ref{lemma3} is the transmission of two correlated sources over a MAC with one distortion criterion, when one source is available at both encoders as considered in \cite{gunduz2007correlated}, which corresponds to $S_2$ being a constant in Theorem~\ref{lemma3}. 
\end{corollary}

A related scenario is when the two sources share a common part in the sense of of G{\'a}cs-K{\"o}rner. 
The following result states that, in accordance with Theorem~\ref{lemma3}, if the two sources are independent when conditioned on the G{\'a}cs-K{\"o}rner common part, then separate source and channel coding is optimal if lossless reconstruction of the common part is required.

\begin{corollary}\label{corollary1new}  
Consider the transmission of correlated sources $S_1$ and $S_2$ with a common part $U_0=f_1(S_1)=f_2(S_2)$ from Definition~\ref{GKW}. 
If $S_1-U_0-S_2$ and the common part $U_0$ of $S_1$ and $S_2$  is to be recovered losslessly,  
then, separate source and channel coding is optimal. 
\end{corollary} 
\begin{proof}
From Definition~\ref{GKW}, the two encoders can separately reconstruct $U_0$. The result then follows by letting $Z\leftarrow U_0$ in Theorem~\ref{lemma3}. 
\end{proof}

\subsubsection{Separation in the Presence of Shared Encoder-Decoder Side Information} 
We next assume that both switches in Fig.~\!\ref{Fig:Model1} are closed, and show the optimality of separation if the two sources are independent given the side information that is shared between the encoders and the decoder. 

\begin{theorem}\label{lemma2}
Consider communication of two correlated sources $S_1$ and $S_2$ with side information $Z$ shared between the encoders and the decoder. If $S_1-Z-S_2$ holds, 
then separation is optimal, and $(D_1, D_2)$ is achievable if 
\begin{align}
R_{S_1|Z} (D_1)&< I(X_1; Y|X_2, Q) \label{lemm2cond1}\\
R_{S_2|Z} (D_2)&< I(X_2; Y|X_1, Q) \label{lemm2cond2} \\
R_{S_1|Z} (D_1) + R_{S_2|Z} (D_2)&< I(X_1, X_2; Y|Q) \label{lemm2cond3}
\end{align} 
for some $p(x_1, x_2, y,q)=p(y|x_1, x_2) p(x_1|q) p(x_2|q)p(q)$. 

Conversely, for any achievable $(D_1, D_2)$ pair, \eqref{lemm2cond1}-\eqref{lemm2cond3} must hold with $<$ replaced with $\leq$.
\end{theorem} 

\begin{proof} 
See Appendix~\ref{appendixA}. 
\end{proof}

When side information $Z$ is available only at the decoder, i.e., when only switch $\text{SW}_1$ is closed, separation is known to be optimal for the lossless transmission of sources $S_1$ and $S_2$ whenever $S_1-Z-S_2$ \cite{gunduz2009source}.  
In light of Theorem~\ref{lemma2}, we show that a similar result holds for the lossy case  whenever 
the Wyner-Ziv rate distortion function of each source is equal to its conditional rate distortion function.

\begin{corollary}\label{lemma2old}
Consider the communication of correlated sources $S_1$ and $S_2$ with decoder only side information $Z$.  If 
\begin{equation}
R_{S_j|Z} (D_j)=R^{\text{WZ}}_{S_j|Z} (D_j), \label{secondcondition}
\end{equation}
where
\begin{equation}
R^{\text{WZ}}_{S_j|Z} (D_j) \triangleq \min_{\substack{p(u_j|s_j), g(u_j, z):\\ \mathbb{E}[d_j(S_j, g(U_j, Z))]\leq D_j\\U_j-S_j-Z}} I(S_j;U_j|Z)  \text{ for } j=1,2, \notag
\end{equation}
is the (Wyner-Ziv) rate distortion function of $S_j$ with decoder-only  side information $Z$ \cite{wyner1976rate}, and  $S_1-Z-S_2$ form a Markov chain, then separation is optimal, with the necessary and sufficient conditions in \eqref{lemm2cond1}-\eqref{lemm2cond3}.  
\end{corollary} 
\begin{proof}
Corollary~\ref{lemma2old} follows from the fact that whenever \eqref{secondcondition} holds, conditional rate distortion functions in Theorem~\ref{lemma2} are achievable by relying on decoder side information only. 
\end{proof}
We note that Gaussian sources are an example for  \eqref{secondcondition}. 
\begin{remark}\label{subopt}
We would like to note that the optimality/sub-optimality of separation for the case of decoder-only side information conditioned on which the two sources are independent is open in general. In addition to the setting in Corollary~\ref{lemma2old}, the optimality of separation holds also for lossless reconstruction \cite{gunduz2009source}. 
\end{remark}

Lastly, we consider the transmissibility of correlated sources with a common part when the common part is available at the decoder. The following result states that if the two sources are independent when conditioned on the G{\'a}cs-K{\"o}rner common part, separation is again optimal if the decoder has access to the common part. 

\begin{corollary}\label{corollary1}  
Consider the transmission of sources $S_1$ and $S_2$ with a common part $U_0=f_1(S_1)=f_2(S_2)$ from Definition~\ref{GKW}. 
Then, separation is optimal if $S_1-U_0-S_2$ 
and the common part $U_0$ is available at the decoder. 
\end{corollary}
\begin{proof} 
Since both encoders can extract $U_0$ individually, each source can achieve the corresponding conditional rate distortion function. Corollary~\ref{corollary1} then follows from Theorem~\ref{lemma2} by letting $Z\leftarrow U_0$. 
\end{proof}
In the following, we leverage these separation results to obtain necessary conditions for the lossy coding of correlated sources over a MAC without side information.

\section{Necessary Conditions for Transmitting Correlated Sources over a MAC} 
We consider in this section the lossy coding of correlated sources over a MAC when both switches in Fig. \ref{Fig:Model1} are open; see Fig.~\ref{Fig:SystemModel100}. We provide necessary conditions for the achievability of a distortion pair $(D_1, D_2)$ using our results from Section~\ref{section3}. This will be achieved by providing correlated side information to the encoders and the decoder, conditioned on which the two sources are independent. From Theorem~\ref{lemma2}, separation is optimal in this setting, and the corresponding necessary and sufficient conditions for the achievability of a distortion pair serve as necessary conditions for the original problem. Corresponding necessary conditions are presented in Theorem~\ref{Thm:Necessary} below.

\begin{theorem}\label{Thm:Necessary} 
Consider the communication of correlated sources $S_1$ and $S_2$ over a MAC.
If a distortion pair $(D_1, D_2)$ is achievable, then for every $Z$ satisfying the Markov condition $S_1-Z-S_2$, we have 
\begin{align}
R_{S_1|Z} (D_1)&\leq I(X_1; Y|X_2, Q), \label{lemm2cond1n}\\
R_{S_2|Z} (D_2)&\leq I(X_2; Y|X_1, Q), \label{lemm2cond2n} \\
R_{S_1|Z} (D_1)  + R_{S_2|Z} (D_2)   &\leq I(X_1, X_2;Y|Q), \label{necesconst4} \\
R_{S_1 S_2} (D_1, D_2)  &\leq I(X_1, X_2;Y), \label{necesconst5}
\end{align}
for some $Q$ for which $X_1-Q-X_2$ form a Markov chain, where 
\begin{equation}
R_{S_1 S_2} (D_1, D_2) =\min_{\substack{p(\hat{s}_1, \hat{s}_2|s_1, s_2)\\ \mathbb{E}[d_1(S_1, \hat{S}_1)]\leq D_1 \\ \mathbb{E}[d_2(S_2, \hat{S}_2)]\leq D_2}} I(S_1, S_2;\hat{S}_1, \hat{S}_2) \notag
\end{equation} 
is the rate distortion function of the joint source $(S_1, S_2)$ with target distortions $D_1$ and $D_2$ for sources $S_1$ and $S_2$, respectively. 

\end{theorem}
\begin{proof}
For any $Z$ that satisfies the Markov condition $S_1-Z-S_2$, we consider the genie-aided setting in which $Z^n$ is provided to  the encoders and the decoder. Then, we obtain the setting in Theorem~\ref{lemma2}. Conditions \eqref{lemm2cond1n}-\eqref{necesconst4}  follow from  Theorem~\ref{lemma2}, whereas condition \eqref{necesconst5} follows from the cut-set bound. 
\end{proof} 

\begin{figure}[t]
\centering
\includegraphics[width=0.65\linewidth]{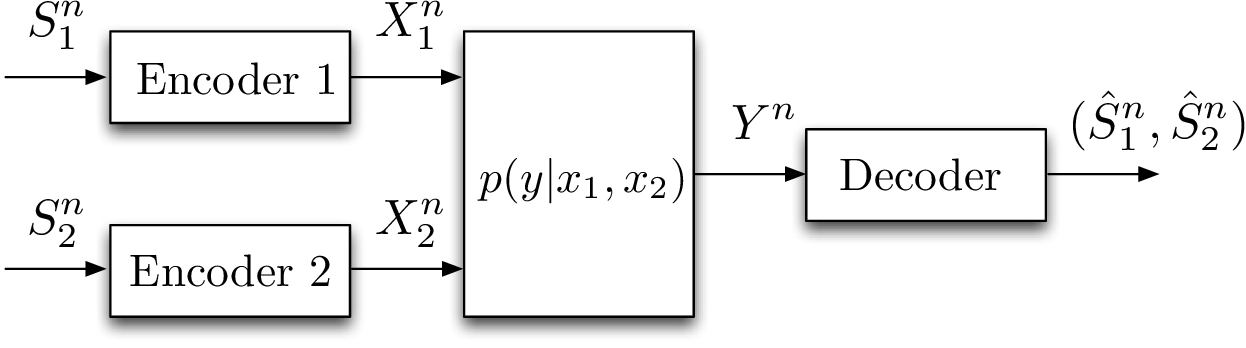}
\caption{Correlated sources over a MAC.}
\label{Fig:SystemModel100}    
\end{figure}

\subsection{Correlated Sources over a Gaussian MAC} 
In this section, we focus on a memoryless MAC with additive Gaussian noise: 
\vspace{-0.19cm}\begin{equation} \label{GausMAC}
Y = X_1 + X_2 + N,
\vspace{-0.11cm}\end{equation}
where $N$ is a standard Gaussian random variable. We impose the input power constraints $\frac{1}{n} \sum_{i=1}^n\mathbb{E}[X_{ji}^2]$ $\leq P$, $j=1,2$. 
In the following, we specialize the necessary conditions of Theorem~\ref{Thm:Necessary} to a Gaussian MAC. 
\begin{corollary}\label{cor6}
If a distortion pair $(D_1, D_2)$ is achievable for sources $(S_1, S_2)$ over the Gaussian MAC in \eqref{GausMAC}, then for every $Z$ that forms a Markov chain $S_1-Z-S_2$, we have
\begin{align} 
R_{S_1|Z} (D_1)  + R_{S_2|Z} (D_2)   &\leq \frac{1}{2} \log ( 1 + \beta_1 P + \beta_2 P ) \label{eq3GGY1}\\
R_{S_1 S_2} (D_1, D_2)  &\leq \frac{1}{2} \log ( 1 + 2P + 2P \sqrt{(1-\beta_1)(1-\beta_2)} ) \label{eq3GGY2}
\end{align}
for some $0\leq \beta_1, \beta_2 \leq 1$. 
\end{corollary}
\begin{proof}
The corollary follows by considering only \eqref{necesconst4}-\eqref{necesconst5}, and from the fact that the right hand sides (RHSs) of these inequalities are maximized by Gaussian $Q$, $X_1$, and $X_2$ \cite{bross2008gaussian}. 
\end{proof}

\subsubsection{Gaussian Sources over a Gaussian MAC} 
This section studies the necessary conditions for transmitting correlated Gaussian sources over a Gaussian MAC.  
Consider a bivariate Gaussian source $(S_1, S_2)$ such that
\begin{equation} \label{bivariate}
\left ( \begin{matrix} S_1 \\S_2\end{matrix}  \right ) \sim
 \mathcal{N}\left ( \left (\begin{matrix} 0 \\0\end{matrix} \right), \left ( \begin{matrix} 1 & \rho \\\rho & 1\end{matrix}\right) \right),
\end{equation}
transmitted over the DM Gaussian MAC in \eqref{GausMAC}, 
under the squared error distortion measures $d_j(S_j, \hat{S}_j)= (S_j-\hat{S}_j)^2$ for $j=1,2$.

For this setup, various notable results exist, each presenting different sets of necessary conditions. 
The following necessary condition is obtained in \cite[Theorem IV.1]{lapidoth2010sending}:
\vspace{-0.1cm}\begin{equation} \label{LT}
R_{S_1 S_2}( D_1, D_2) \leq \frac{1}{2} \log (1 + 2P(1+\rho)).
\vspace{-0.1cm}\end{equation}  
Another set of necessary conditions is proposed in \cite[Theorem 2]{jain2012energy}. By substituting $\sigma_Z^2=\sigma_1^2=\sigma_2^2=1$ and $E_1=E_2=P$  in  \cite[Theorem 2]{jain2012energy}, these conditions can be stated as follows:
\begin{align}
\frac{1}{(1-\hat{\rho})^2} \ln \left( \frac{1-\rho^2}{D_k}\right) &\leq P,  \quad k=1,2, \label{nc1}\\
(\ln2) R_{S_1 S_2}(D_1, D_2) &\leq P(1+\hat{\rho}), \label{nc2}
\end{align}
for some $0\leq \hat{\rho}\leq |\rho|$. 

Other sets of necessary conditions have recently been presented in \cite[Theorem 1]{7541654},  \cite[Proposition 2]{tian2017matched}, and \cite[Theorems 1 and 4]{yu2016distortion}, all incorporating various auxiliary random variables. It is not possible in general to compare Theorem~\ref{Thm:Necessary} over the full set of conditions presented in these results, since this involves optimization of auxiliary random variables and a large number of parameters. For this reason, here we compare Corollary~\ref{cor6} with  \eqref{LT},  \eqref{nc1}-\eqref{nc2}, along with the conditions from \cite[Corollary 1.1]{7541654}, which is a relaxed version of \cite[Theorem 1]{7541654}. Note that Corollary~\ref{cor6} is also a weaker version of  Theorem~\ref{Thm:Necessary}, where, for fairness, the first two single rate conditions are removed as in \cite[Corollary 1.1]{7541654}.

The set of necessary conditions from \cite[Corollary 1.1]{7541654} can be stated as:
\begin{align}
R_{S_1S_2}(D_1, D_2)  - \frac{1}{2} \log \frac{1+\rho}{1-\rho}  &\leq \frac{1}{2} \log ( 1 + \beta_1 P + \beta_2 P )  \label{eq3LW1}\\
R_{S_1 S_2} (D_1, D_2) &\leq \frac{1}{2} \log ( 1 + 2P + 2P \sqrt{(1-\beta_1)(1-\beta_2)} )  \label{eq3LW2}
\end{align}
for some $0\leq \beta_1, \beta_2 \leq 1$.

For the necessary conditions in Corollary~\ref{cor6}, we let $Z$  be the common part of $(S_1, S_2)$ with respect to Wyner's common information from \eqref{commoninfo}. 
The common part can be characterized as follows \cite[Proposition 1]{xu2016lossy}. Let $Z$, $N_1$, and $N_2$ be standard random variables. Then, $S_1$, and $S_2$ can be expressed as 
\begin{equation}
S_i = \sqrt{\rho} Z + \sqrt{1-\rho} N_i, \quad i=1,2,
\end{equation}
where $I(S_1, S_2 ; Z) = \frac{1}{2} \log \frac{1+\rho}{1-\rho}$ and $I(S_1, S_2; Z') > \frac{1}{2} \log \frac{1+\rho}{1-\rho}$ for all $S_1-Z'-S_2$ with $Z'\neq Z$.

The rate distortion function for $S_i$ with encoder and decoder side information $Z$ is \cite{wyner1978rate}: 
\begin{equation}\label{point1}
R_{S_i|Z}(D_i)  = \left \{ \begin{matrix} 
\frac{1}{2} \log \frac{1-\rho}{D_i} & \text{ if } & \qquad 0 < D_i < 1-\rho \\
0 & \text{ if } & D_i \geq  1-\rho 
\end{matrix}\right .  
\end{equation}
for $i=1,2$. 
We also have, from \cite{xiao2005compression, lapidoth2010sending}, that,
\begin{align} \label{RateGausLHS}
R_{S_1 S_2} (D_1, D_2)  = \left\{ \begin{matrix}
\frac{1}{2} \log \left( \frac{1}{\min(D_1, D_2)}\right ) & \text{ if } (D_1, D_2) \in \mathcal{D}_1 \\
\frac{1}{2} \log^+ \left( \frac{1-\rho^2}{D_1 D_2}\right ) & \text{ if } (D_1, D_2) \in \mathcal{D}_2 \\
\frac{1}{2} \log^+ \left( \frac{1-\rho^2}{D_1 D_2 - \left(\rho - \sqrt{(1-D_1)(1-D_2)}\right)^2}\right ) & \text{ if } (D_1, D_2) \in \mathcal{D}_3 \\
\end{matrix}\right . , 
\end{align} 
where $\log^+ (x) = \max\{0, \log(x)\}$, and 
\begin{align}
\mathcal{D}_1&=\bigg\{  
(D_1, D_2): (0\leq D_1 \leq 1-\rho^2, D_2 \geq 1-\rho^2 + \rho^2D_1) \text{ or } \nonumber \\ 
&\qquad \qquad \qquad \qquad \Big(1-\rho^2 < D_1 \leq 1, D_2 \geq 1-\rho^2 + \rho^2 D_1, D_2 \leq \frac{D_1-(1-\rho^2)}{\rho^2}\Big)
\bigg\} \\
\mathcal{D}_2&=\bigg\{  
(D_1, D_2): 0\leq D_1 \leq 1-\rho^2, 0\leq D_2 < (1-\rho^2-D_1) \frac{1}{1-D_1}
\bigg\} \\
\mathcal{D}_3&=\bigg\{  
(D_1, D_2): \Big(0\leq D_1 \leq 1-\rho^2, (1-\rho^2-D_1) \frac{1}{1-D_1} \leq D_2 < 1-\rho^2 + \rho^2 D_1 \Big) \text{ or } \nonumber \\  
&\qquad \qquad \qquad \qquad \qquad \Big( 
1-\rho^2 < D_1 \leq 1, \frac{D_1-(1-\rho^2)}{\rho^2} < D_2 < 1-\rho^2 +\! \rho^2 D_1\Big)
\bigg\} .
\end{align}
Fig.~\ref{Fig:1-regions} illustrates the regions $\mathcal{D}_1$, $\mathcal{D}_2$, and $\mathcal{D}_3$ as in \cite{lapidoth2010sending}.

\begin{figure}[t]
\centering
\subfloat[]{\includegraphics[width=0.35\linewidth]{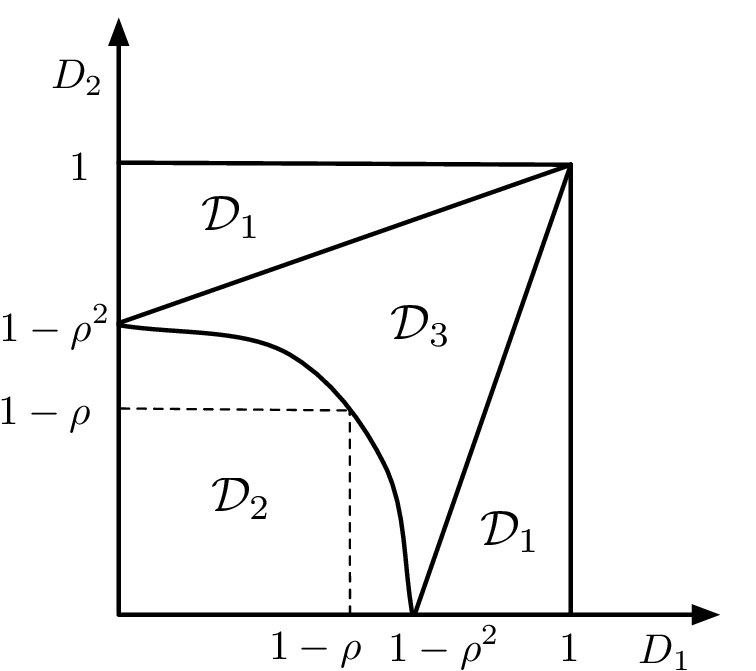}\label{Fig:1-regions}}\hspace{2cm} 
\subfloat[]{\includegraphics[width=0.35\linewidth]{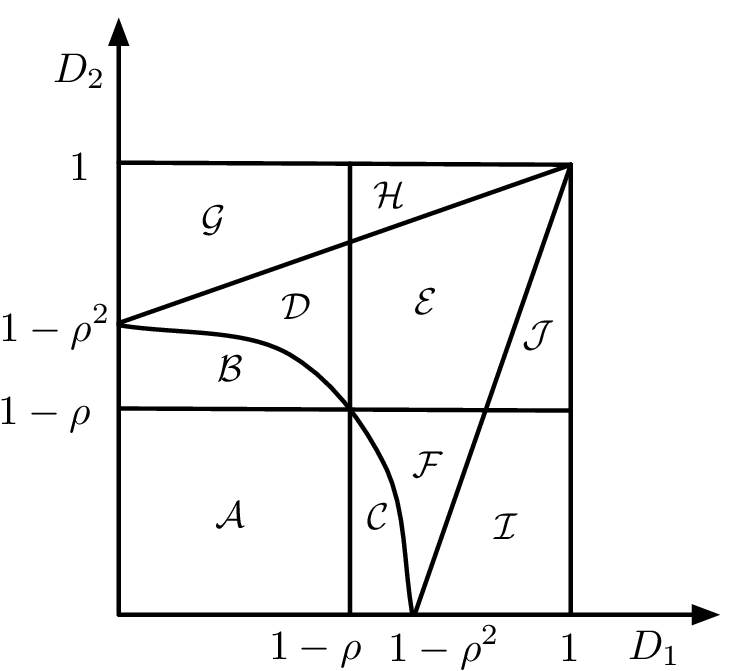}\label{regions2}} 
\caption{\protect\subref{Fig:1-regions} Regions $\mathcal{D}_1$, $\mathcal{D}_2$, and $\mathcal{D}_3$. \protect\subref{regions2} Partitioned distortion regions for $(D_1, D_2)$.}
\label{fig:distortionregions}
\end{figure}

By analyzing the corresponding expressions from  Corollary~\ref{cor6},  \eqref{LT}, \eqref{nc1}-\eqref{nc2}, and \eqref{eq3LW1}-\eqref{eq3LW2}, the next proposition shows that there exist $(D_1, D_2)$ values 
for which Corollary~\ref{cor6} is tighter; that is, while other results cannot make any judgement on the achievability of such $(D_1, D_2)$ pairs, they are shown not to be achievable thanks to Corollary~\ref{cor6}. 
\begin{proposition}\label{propbounds1}
There exist distortion pairs that are included in the outer bounds of \cite[Theorem IV.1]{lapidoth2010sending}, \cite[Theorem 2]{jain2012energy}, and \cite[Corollary 1.1]{7541654}, but not in the outer bound of Corollary~\ref{cor6}. 
\end{proposition}
\begin{proof}
The details are given in Appendix~\ref{appendixD}. 
\end{proof}

A graphical illustration of the bounds from Corollary~\ref{cor6}, \cite[Theorem IV.1]{lapidoth2010sending}, and \cite[Corollary 1.1]{7541654} can be provided as follows.  
Define
\begin{align}
r_1 (\beta_1, \beta_2) &\triangleq \frac{1}{2} \log ( 1 + 2P + 2P \sqrt{(1-\beta_1)(1-\beta_2)} ),   
\\
r_2 (\beta_1, \beta_2)  &\triangleq \frac{1}{2} \log ( 1 + \beta_1 P + \beta_2 P ), 
\end{align}
and consider the region 
\begin{equation}\label{region}
\mathcal{R} = \bigcup_{0\leq \beta_1, \beta_2\leq 1} \left \{(R_1, R_2): R_1\leq r_1(\beta_1, \beta_2), R_2\leq r_2(\beta_1, \beta_2) \right \}. 
\end{equation} 
The necessary conditions in Corollary~\ref{cor6} state that, if a $(D_1, D_2)$ pair is achievable, then 
\begin{equation}\label{checkGGY}
\left (R_{S_1 S_2} (D_1, D_2), R_{S_1|Z} (D_1)  + R_{S_2|Z} (D_2) \right ) \in \mathcal{R}. 
\end{equation} 
The necessary conditions in \eqref{eq3LW1}-\eqref{eq3LW2} state that, if a $(D_1, D_2)$ pair is achievable, then
\begin{equation}\label{checkLW}
\left(R_{S_1 S_2} (D_1, D_2), R_{S_1 S_2} (D_1, D_2) - \frac{1}{2} \log \frac{1+\rho}{1-\rho}  \right) \in \mathcal{R}.
\end{equation}
Let $D_1 = 0.145 <  1-\rho$. 
Consider first Region $\mathcal{B}$, for which $D_1 \leq 1-\rho$ and $1-\rho \leq D_2 \leq \frac{1-\rho^2-D_1}{1-D_1}$. 
For a $(D_1, D_2)$ pair in Region $\mathcal{B}$, i.e., $D_1=0.145$ and $1-\rho\leq D_2\leq \frac{1-\rho^2-D_1}{1-D_1}$, we have from \eqref{point1} and \eqref{RateGausLHS} that
\begin{equation}\label{pointDDYcheck}
\left ( R_{S_1 S_2} (D_1, D_2) , R_{S_1|Z} (D_1)  + R_{S_2|Z} (D_2)\right ) = \left (\frac{1}{2}\log \frac{1-\rho^2}{D_1D_2}, \frac{1}{2}\log \frac{1-\rho}{D_1} \right ). 
\end{equation}
The $(R_{S_1 S_2}(D_1, D_2), R_{S_1|Z}(D_1)+R_{S_2|Z}(D_2))$ pairs obtained from \eqref{pointDDYcheck} for increasing $D_2$ values within Region $\mathcal{B}$ are illustrated with a green ``+'' sign in Fig.~\ref{fig:p065}. The region $\mathcal{R}$ from \eqref{region} is the region shaded in blue in the same figure. Whenever a point from $\eqref{pointDDYcheck}$ falls outside of $\mathcal{R}$, we conclude that the corresponding $(D_1, D_2)$ pair is not achievable according to Corollary~\ref{cor6}. 
\begin{figure}[t]
\centering
\subfloat[]{\includegraphics[width=0.5\linewidth]{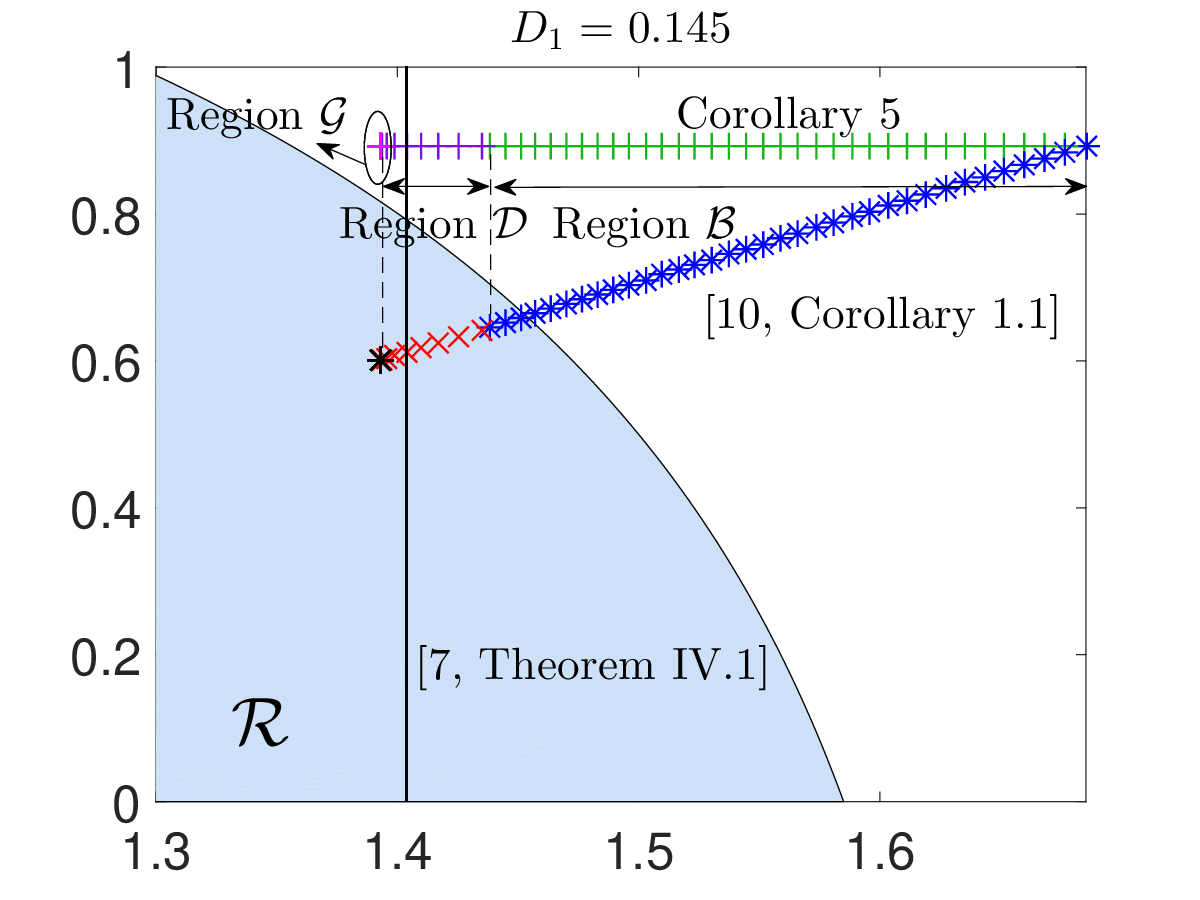}\label{fig:p065}} 
\subfloat[]{\includegraphics[width=0.5\linewidth]{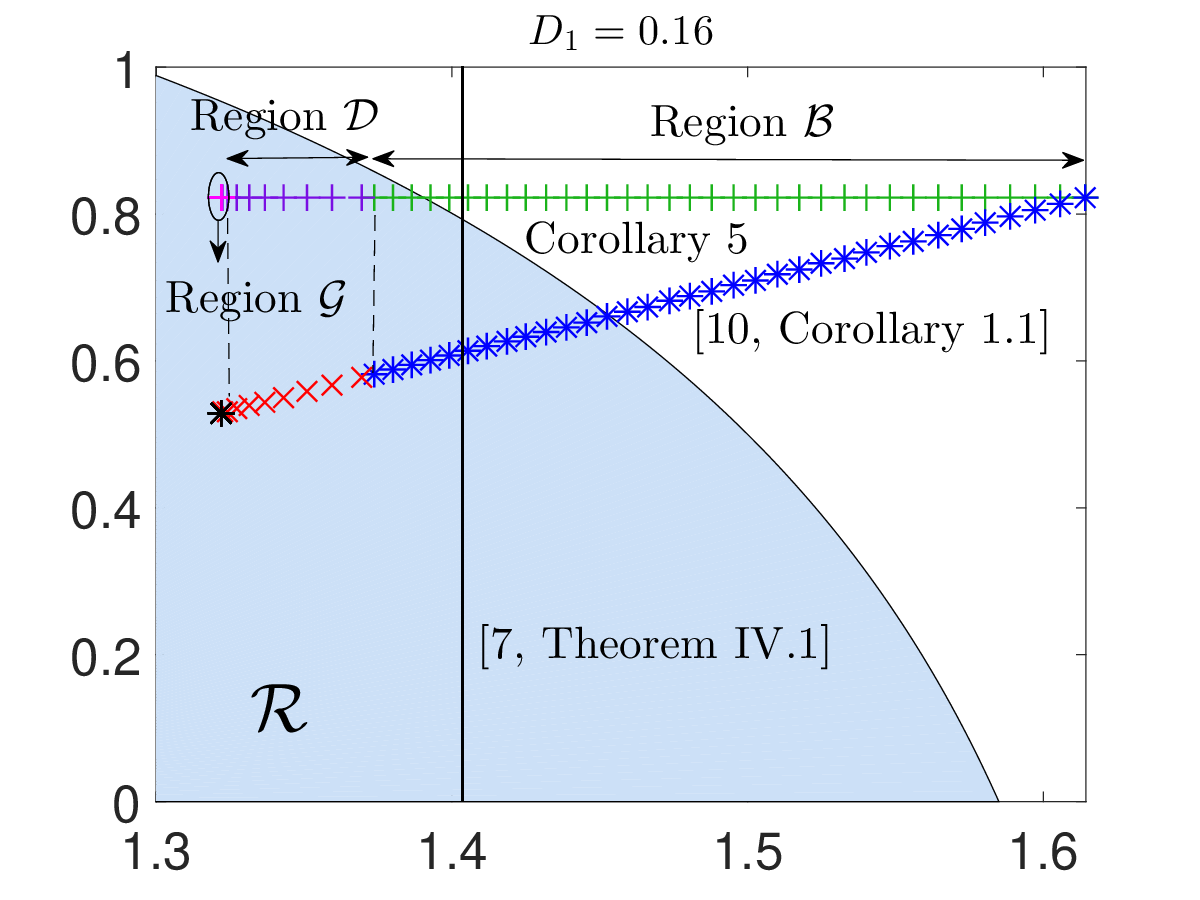}\label{fig:p16}} 
\caption{Comparison of the necessary conditions from Corollary~\ref{cor6} with the necessary conditions from \eqref{LT}  and \eqref{eq3LW1}-\eqref{eq3LW2}, respectively, for $P=2$, $\rho=0.5$, and \protect\subref{fig:p065} $D_1=0.145$, \protect\subref{fig:p16} $D_1=0.16$.}
\label{fig:dists}
\end{figure}
We also evaluate
\begin{equation}\label{pointLWheck}
\left(R_{S_1 S_2} (D_1, D_2), R_{S_1 S_2} (D_1, D_2) - \frac{1}{2} \log \frac{1+\rho}{1-\rho}  \right) = \left (\frac{1}{2}\log \frac{1-\rho^2}{D_1D_2}  , \frac{1}{2}\log \frac{(1-\rho)^2}{D_1D_2}   \right )
\end{equation}
for points $(0.145, D_2)$ in Region $\mathcal{B}$, using \eqref{RateGausLHS}. The points corresponding to \eqref{pointLWheck} for different $D_2$ values are marked with a dark blue ``*'' in Fig.~\ref{fig:p065}. 
Whenever a point from $\eqref{pointLWheck}$ is not contained within $\mathcal{R}$, then the corresponding $(D_1, D_2)$ pair is not achievable according to \eqref{eq3LW1}-\eqref{eq3LW2}.

Next, we consider $(D_1, D_2)$ pairs from Region $\mathcal{D}$, for which $D_1 \leq 1-\rho$ and $\frac{1-\rho^2-D_1}{1-D_1}\leq D_2 \leq 1-\rho^2+\rho^2D_1$. We evaluate 
\begin{align}\label{pointDDYcheck1}
\left (\!R_{S_1 S_2} (D_1, \!D_2), R_{S_1|Z} (D_1)  \!+\! R_{S_2|Z} (\!D_2) \!\right ) \!=\!\! \left (\!\frac{1}{2} \log^+\!\! \bigg( \frac{1\!-\!\rho^2}{D_1 D_2 \!-\! \big(\rho \!-\!\! \sqrt{\!(1\!\!-\!\!D_1)(1\!\!-\!\!D_2)}\big)^2}\!\bigg), \!\frac{1}{2}\log \!\frac{1\!-\!\rho}{D_1} \!\right )
\end{align}
from \eqref{point1}-\eqref{RateGausLHS}. The values obtained for $D_1=0.145$ and $D_2\in \left(\frac{1-\rho^2-D_1}{1-D_1}, 1-\rho^2+\rho^2D_1\right)$ are marked with a purple ``+'' in Fig.~\ref{fig:p065}. Similarly, from \eqref{RateGausLHS}, for $(D_1, D_2) \in \text{Region }\mathcal{D}$, 
\vspace{-0.2cm}\begin{align}\label{pointLWheck1}
&\left(R_{S_1 S_2} (D_1, D_2), R_{S_1 S_2} (D_1, D_2) - \frac{1}{2} \log \frac{1+\rho}{1-\rho}  \right) \notag \\
& = \!\left (\!
\frac{1}{2} \log^+\!\! \Bigg( \!\frac{1\!-\!\rho^2}{D_1 D_2 \!-\! \big(\rho \!-\!\! \sqrt{\!(1\!\!-\!\!D_1)(1\!\!-\!\!D_2)}\big)^2}\Bigg), \!
\frac{1}{2} \log^+\!\! \Bigg(\! \frac{(1\!-\!\rho)^2}{D_1 D_2 \!-\! \big(\rho \!-\! \sqrt{\!(1\!\!-\!\!D_1)(1\!\!-\!\!D_2)}\big)^2}\!\Bigg)\!\!\!
 \right ) \!. 
\end{align}
Corresponding points for $D_1=0.145$ and increasing $D_2$ values in Region $\mathcal{D}$ are illustrated with a red ``x'' marking in Fig.~\ref{fig:p065}.

Finally, we consider $(D_1, D_2)\in \text{Region } \mathcal{G}$, where $D_1 \leq 1-\rho$, $1-\rho^2+\rho^2D_1\leq D_2\leq 1$, and
\begin{equation}\label{pointDDYcheck12}
\left (R_{S_1 S_2} (D_1, D_2), R_{S_1|Z} (D_1)  + R_{S_2|Z} (D_2) \right ) 
= \left(\frac{1}{2} \log  \frac{1}{D_1},  \frac{1}{2} \log \frac{1-\rho}{D_1}\right ). 
\end{equation}
Corresponding points are marked with a pink ``+'' in Fig.~\ref{fig:p065}. 
Note that since \eqref{pointDDYcheck12} depends only on $D_1$, these points appear as a single point.  
We also evaluate
\begin{equation}\label{pointLWheck12}
\left(R_{S_1 S_2} (D_1, D_2), R_{S_1 S_2} (D_1, D_2) - \frac{1}{2} \log \frac{1+\rho}{1-\rho}  \right)  = 
\left(\frac{1}{2} \log  \frac{1}{D_1}, \frac{1}{2} \log  \frac{1-\rho}{D_1(1+\rho)}\right)
\end{equation}
for $1-\rho^2+\rho^2D_1\leq D_2\leq 1$ from \eqref{RateGausLHS}. This is marked with a black ``*'' in Fig.~\ref{fig:p065}. Since \eqref{pointLWheck12} also depends only on $D_1$, they appear as a single point.  
One can observe from \eqref{pointDDYcheck}-\eqref{pointLWheck}, as well as from \eqref{pointDDYcheck1}-\eqref{pointLWheck1} and \eqref{pointDDYcheck12}-\eqref{pointLWheck12},  that the points that share the same value on the horizontal axis in Fig.~\ref{fig:p065}  correspond to the same $(D_1, D_2)$ pairs, as the first terms of both \eqref{pointDDYcheck}-\eqref{pointLWheck} and  \eqref{pointDDYcheck1}-\eqref{pointLWheck1} as well as \eqref{pointDDYcheck12}-\eqref{pointLWheck12} are equal.

Lastly, we illustrate the RHS of \eqref{LT} with a straight line in Fig.~\ref{fig:p065}. The points on the RHS of this line correspond to $(D_1, D_2)$ pairs that are not achievable according to \eqref{LT}, since for these points one has
\begin{equation} 
R_{S_1 S_2}( D_1, D_2) > \frac{1}{2} \log (1 + 2P(1+\rho)). 
\end{equation} 

In order to compare the three bounds, we investigate the $(D_1, D_2)$ pairs that cannot be achieved by Corollary~\ref{cor6}, \eqref{LT}, and \eqref{eq3LW1}-\eqref{eq3LW2}, respectively. 
From Fig.~\ref{fig:p065}, we find that when $D_1=0.145$, some $(D_1, D_2)$ pairs in Regions $\mathcal{G}$ and $\mathcal{D}$ (from Fig.~\ref{regions2}) satisfy both \eqref{LT} and \eqref{eq3LW1}-\eqref{eq3LW2}, but not Corollary~\ref{cor6}, as can be observed from the pink and purple points marked with the ``+'' sign that are on the left hand side (LHS) of the straight line,  but outside of $\mathcal{R}$. 
Therefore, we can conclude that there exist distortion pairs for which Corollary~\ref{cor6} provides tighter conditions than both \eqref{LT} and \eqref{eq3LW1}-\eqref{eq3LW2} in Regions $\mathcal{G}$ and $\mathcal{D}$.

We also compare the corresponding bounds when $D_1=0.16$ in Fig.~\ref{fig:p16}. From the green points marked with the ``+'' sign that are on the LHS of the straight line but are outside of $\mathcal{R}$, we observe that there exist distortion pairs in Region $\mathcal{B}$ for which Corollary~\ref{cor6} provides tighter conditions than both \eqref{LT} and \eqref{eq3LW1}-\eqref{eq3LW2}.


We note, however, that Corollary~\ref{cor6} is not necessarily strictly tighter for all $(D_1, D_2)$ pairs. The next proposition states that there exist $(D_1, D_2)$ pairs for which \eqref{LT} is tighter than Corollary~\ref{cor6}. 

\begin{proposition}\label{propbounds2}
There exist distortion pairs that are in the outer bound of Corollary~\ref{cor6}, but not in the outer bound of \cite[Theorem IV.1]{lapidoth2010sending}. 
\end{proposition}
\begin{proof}
The details are available in Appendix~\ref{appendixE}. 
\end{proof}

\subsubsection{Binary Sources over a Gaussian MAC} 
We next study  the transmission of a doubly symmetric binary source (DSBS) over a Gaussian MAC. Consider a DSBS with joint distribution
\vspace{-0cm}\begin{equation}
p(S_1\!=\!s_1, S_2\!=\!s_2) = \frac{1\!-\!\alpha}{2} (1\!-\!|s_1\!-\!s_2|)  +  \frac{\alpha}{2} |s_1\!-\!s_2|,
\vspace{-0cm}\end{equation}
a memoryless Gaussian MAC from \eqref{GausMAC}, 
and Hamming distortion $d_j(S_j, \hat{S}_j)\!=\! |S_j \!-\! \hat{S}_j|$ where $\hat{\mathcal{S}}_j\!=\!\mathcal{S}_j\!=\!\{0,1\}$ for $j=1,2$.

For the conditions in Corollary~\ref{cor6}, we choose the variable $Z$ as illustrated in Fig.~\ref{Fig:Z}a. Then the joint distribution for $(S_i, Z)$ is as given in Fig.~\ref{Fig:Z}b for $i=1,2$. Note that $Z$ forms a $Z$-channel both with $S_1$ and $S_2$ while satisfying $S_1-Z-S_2$. Using the conditional rate-distortion function for the $Z$-channel setting from \cite{steinberg2009coding}, one can evaluate Corollary~\ref{cor6}. 

We compare Corollary~\ref{cor6} first with the set of necessary conditions from \cite[Remark IV.1]{lapidoth2010sending},
\begin{equation}\label{LTDSBS}
R_{S_1 S_2}( D_1, D_2) \leq \frac{1}{2} \log (1 + 2P(1+\rho_{max})),
\end{equation}
where $R_{S_1 S_2}( D_1, D_2)$ is as in \cite[Theorem 2]{nayak2010successive}, and $\rho_{max}$ is the Hirschfield-Gebelin-R\'{e}nyi maximal correlation for DSBS given by  \cite{anantharam2014hypercontractivity}:
\begin{equation}
\rho_{max} = \sqrt{2(\alpha^2 + (1-\alpha)^2) -1}.
\end{equation}

\begin{figure}[t]
\centering
\includegraphics[width=0.65\linewidth]{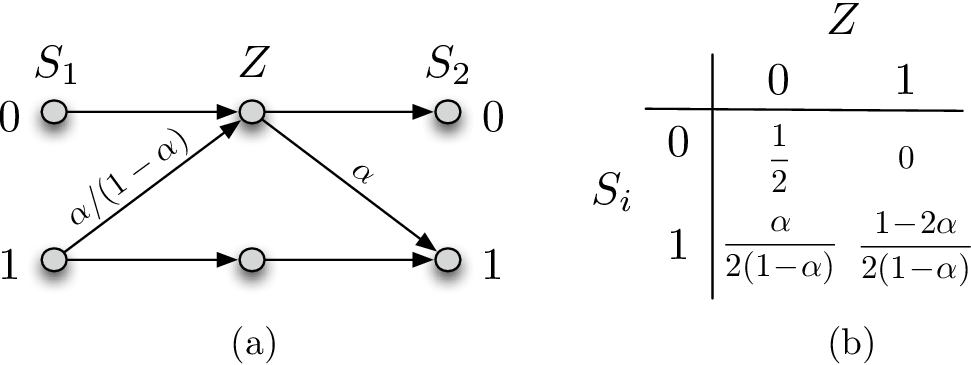}
\caption{(a) Z-channel structure. (b) $p(S_i, Z)$ for $i=1,2$.}
\label{Fig:Z}    
\end{figure}

We next consider the necessary conditions from \cite[Corollary 1.1]{7541654},
\begin{align}
&R_{S_1S_2}(D_1, D_2)  - 1 - h(\alpha) + 2h(\theta) \leq \frac{1}{2} \log ( 1 + \beta_1 P + \beta_2 P )  \label{LWDS11}\\
&R_{S_1 S_2} (D_1, D_2) \leq \frac{1}{2} \log ( 1 + 2P + 2P \sqrt{(1-\beta_1)(1-\beta_2)} )  \label{LWDS12}
\end{align} 
for some $0\leq \beta_1, \beta_2 \leq 1$, where $\theta = (1/2)(1-\sqrt{1-2\alpha})$ and $h(\lambda)=-\lambda\log\lambda-(1-\lambda)\log(1-\lambda)$ is the binary entropy function, and $C_{W}(S_1, S_2)$ from \eqref{commoninfo} is as in   \cite{wyner1975common}.

The last set of necessary conditions we consider is obtained from \cite[Theorem 1]{7541654} by removing (9a) and (9b) and letting $W \leftarrow Z$, where $Z$ is as defined in Fig. \ref{Fig:Z}, 
\begin{align}
&R_{S_1S_2}(D_1, D_2)  - 1+\frac{\alpha}{1-\alpha}h(\alpha)\leq \frac{1}{2} \log ( 1 + \beta_1 P + \beta_2 P ),  \label{LWDS21} \\
&R_{S_1 S_2} (D_1, D_2) \leq \frac{1}{2} \log ( 1 + 2P + 2P \sqrt{(1-\beta_1)(1-\beta_2)} ),  \label{LWDS22}
\end{align}
for some $0\leq \beta_1, \beta_2 \leq 1$. 
In the following, we compare Corollary~\ref{cor6} 
with the necessary conditions from \eqref{LTDSBS} and \eqref{LWDS11}-\eqref{LWDS12} as well as from \eqref{LWDS21}-\eqref{LWDS22}. 

\begin{proposition}\label{propbound3}
There exist distortion pairs  that satisfy the outer bounds of \cite[Remark IV.1]{lapidoth2010sending}, \cite[Corollary 1.1]{7541654}, and \eqref{LWDS21}-\eqref{LWDS22} but not the outer bound of Corollary~\ref{cor6} for the binary setup. 
\end{proposition}
\begin{proof}
The details are provided in Appendix~\ref{appendixF}.
\end{proof}

\section{Conclusions}\label{Sec:Conc} 
We have considered the lossy coding of correlated sources over a MAC. We have provided an achievable scheme for the transmission of correlated sources in the presence of decoder side information, and investigated the conditions under which separate source and channel coding is optimal when the encoder and/or decoder has access to side information.  
By leveraging the obtained separation theorem in the presence of a common side information conditioned on which the two sources are independent, we derived a simple and computable set of necessary conditions for the lossy coding of correlated sources over a MAC. 
The comparison of the new necessary conditions with the known results from the literature are provided for the Gaussian setting, i.e., Gaussian sources transmitted over a Gaussian MAC, as well as for a DSBS over a Gaussian MAC. 
Identifying necessary conditions for the transmissibility of correlated sources is an active open research direction. 
A direct comparison of the proposed necessary conditions appear to be difficult analytically, and, due to the dimensionality of the search space, numerically. 
Accordingly, we point to this problem as an interesting future direction. Another interesting open problem is the optimality/suboptimality of separation in the presence of decoder-only side information, conditioned on which the two sources are independent. 
Other future directions include the (sub)optimality of separation in other multi-terminal scenarios with side information.

\appendices 

\section{Proof of Theorem~\ref{lemma:hybrid}}\label{appendix0}
Our achievable scheme is along the lines of \cite{minero2015unified}. For completeness, we provide the details in the sequel. 

\emph{Generation of the codebook:} Choose $\epsilon > \epsilon' > 0$. 
Fix $p(u_1|s_1)$, $p(u_2|s_2)$, $x_1(u_1, s_1)$, $x_2(u_2, s_2)$, $\hat{s}_1(u_1, u_2, y, z)$ and $\hat{s}_2(u_1, u_2, y, z)$ with $\mathbb{E}[d_j(S_j, \hat{S}_j)]\leq \frac{D_j}{1+\epsilon}$  for $j=1,2$. 

For each $j=1,2$, generate $2^{nR_j}$ sequences $u_j^n(m_j)$ for $m_j\in \{1,\ldots,  2^{nR_j}\}$  independently at random conditioned on the distribution $\prod_{i=1}^n p_{U_j}(u_{ji})$. The codebook is known by the two encoders and the decoder. 

\emph{Encoding:}
Encoder $j=1,2$ observes a sequence $s_j^n$ and tries to find an index $m_j\in \{1, \ldots,  2^{nR_j}\}$ such that the corresponding $u_j^n(m_j)$ is jointly typical with $s_j^n$, i.e., $(s_j^n, u_j^n(m_j))\in \mathcal{T}_{\epsilon'}^{(n)}$. If more than one index exist, the encoder selects one of them uniformly at random. If no such index exists, it selects a random index uniformly. Upon selecting the index, encoder~$j$ sends $x_{ji}=x_j (u_{ji}(m_j), s_{ji})$ for $i=1,\ldots, n$ to the decoder.

\emph{Decoding:} The decoder observes the channel output $y^n$ and side information $z^n$, and tries to find a unique pair of indices $(\hat{m}_1, \hat{m}_2)$ such that $(u_1^n(\hat{m}_1), u_2^n(\hat{m}_2), y^n, z^n)\in \mathcal{T}_{\epsilon}^{(n)}$ and sets $\hat{s}_{ji}=\hat{s}_j(u_{1i}(m_1), u_{2i}(m_2), y_i, z_i)$ for $i=1, \ldots, n$ for $j=1, 2$. 

\emph{Expected Distortion Analysis:} 
Let $M_1$ and $M_2$ denote the indices selected by encoder $1$ and encoder $2$.  Define 
\begin{equation}
\mathcal{E} \{(S_1^n, S_2^n, U_1^n(\hat{M}_1), U_2^n(\hat{M}_2), Y^n, Z^n)\notin \mathcal{T}_{\epsilon}^{(n)}\}
\end{equation}
such that the distortion pair $(D_1, D_2)$ is satisfied if $P(\mathcal{E})\rightarrow 0$ as $n\rightarrow\infty$. Let
\begin{align}
\mathcal{E}_j& = \{(S_j^n, U_j^n(m_j))\notin \mathcal{T}_{\epsilon'}^{(n)} \;\;\forall m_j\} , \quad j=1,2\\
\mathcal{E}_3& = \{(S_1^n, S_2^n, U_1^n(M_1), U_2^n(M_2), Y^n, Z^n)\notin \mathcal{T}_{\epsilon}^{(n)}\} \\
\mathcal{E}_4& = \{(U_1^n(m_1), U_2^n(m_2), Y^n, Z^n)\in \mathcal{T}_{\epsilon}^{(n)} \text{ for some } m_1\neq M_1, m_2\neq M_2\} \\
\mathcal{E}_5& = \{(U_1^n(m_1), U_2^n(M_2), Y^n  , Z^n) \in  \mathcal{T}_{\epsilon}^{(n)} \text{  for some } m_1  \neq  M_1\} \label{error25}\\
\mathcal{E}_6& = \{(U_1^n(M_1), U_2^n(m_2), Y^n  , Z^n) \in  \mathcal{T}_{\epsilon}^{(n)} \text{  for some } m_2  \neq  M_2\} 
\end{align}
Then, 
\begin{align}
&P(\mathcal{E})\leq P(\mathcal{E}_1) + P(\mathcal{E}_2) + P(\mathcal{E}_3 \cap \mathcal{E}_1^c \cap \mathcal{E}_2^c) + P(\mathcal{E}_4) + P(\mathcal{E}_5) + P(\mathcal{E}_6).
\end{align} 
Through standard techniques based on joint typicality coding, it can be shown that $P(\mathcal{E})\rightarrow 0$ as $n\rightarrow \infty$ and one can bound the expected distortions for $\mathcal{E}^c$ for the two sources $S_1$ and $S_2$,  when the sufficient conditions in \eqref{eq4n}-\eqref{eq6n} are satisfied.

\section{Proof of Theorem~\ref{lemma3}}\label{appendixB}

\subsection{Achievability}
Our source coding part is based on the distributed source coding scheme with a common part from \cite{wagner2011distributed}. For completeness, we briefly outline the problem setup in \cite{wagner2011distributed}, also depicted in Fig.~\ref{Fig:SystemModelWKA}. This problem considers the transmission of correlated DM sources $(Y_0, Y_1, Y_2)$ such that $Y_j$ is observed by encoder $j=0,1,2$. Lossless reconstruction of source $Y_0$ is required at the decoder, while the remaining two sources, $Y_1$ and $Y_2$, are recovered in a lossy manner, with respect to corresponding per-letter distortion constraints. In other words, we have
\begin{equation} \label{perletter}
\limsup_{n\rightarrow \infty} \frac{1}{n}\sum_{i=1}^n\mathbb{E}[d_j(Y_{ji}, \hat{Y}_{ji})]  \leq D_j, \quad j=1,2.
\end{equation}
and $P(Y_0^n\neq \hat{Y}_0^n)\rightarrow 0$ as $n\rightarrow \infty$. Sources $Y_1$ and $Y_2$ also have a common component $X$ such that, for a pair of deterministic functions $f$ and $g$, $X=f(Y_1) = g(Y_2)$ 
and $H(X)>0$. 
An achievable rate-distortion region for the distributed source coding system in Fig.~\ref{Fig:SystemModelWKA} is given in \cite[Theorem~$1$]{wagner2011distributed}.

\begin{figure}[t]
\centering
\includegraphics[width=0.5\linewidth]{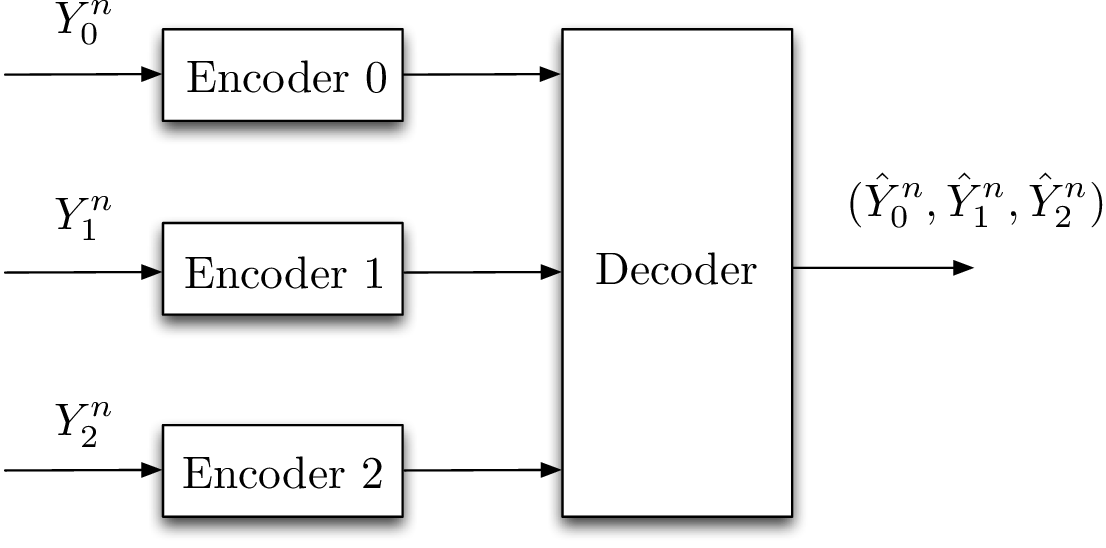}
\caption{Distributed source coding for correlated sources $(Y_0, Y_1, Y_2)$ where $Y_j$ is observed by encoder $j=0,1,2$. The decoder reconstructs $Y_0$ losslessly, while $Y_1$ and $Y_2$ are reconstructed in a lossy manner, with respect to the distortion criterion in \eqref{perletter}.}
\label{Fig:SystemModelWKA}    
\end{figure}

By letting $Y_0\leftarrow Z$, $Y_j\leftarrow (S_j, Z)$ for $j=1,2$, and $X\leftarrow Z$ in Fig.~\ref{Fig:SystemModelWKA}, we observe for this setup that any achievable rate pair for the system in Fig.~\ref{Fig:SystemModelWKA} is also achievable for our system. This follows from the fact that in our setup $Z$ is available to both encoders, as a result, the encoders can cooperate to send it to the decoder and realize any achievable scheme in \cite{wagner2011distributed}.

Letting $U= X$ in \cite[Theorem~$1$]{wagner2011distributed} and substituting $X\leftarrow Z$, $Y_0\leftarrow Z$, $\hat{Y}_0\leftarrow \hat{Z}$, $Y_j\leftarrow (S_j, Z)$, $V_j\leftarrow U_j$, $\hat{Y}_j\leftarrow \hat{S}_j$, and $d_j(Y_j, \hat{Y}_j)\leftarrow d_j(S_j, \hat{S}_j)$ for $j=1,2$, we find that a distortion pair $(D_1, D_2)$ is achievable for the rate triplet $(R_0, R_1, R_2)$ if 
\begin{align}
R_0&\geq H(Z|Z, U_1, U_2) \label{eqCommon1}\\
R_1&\geq I(S_1, Z; U_1|Z, U_2) \label{eqCommon2}\\
R_2&\geq I(S_2, Z; U_2|Z, U_1) \label{eqCommon3}\\
R_0+R_1&\geq H(Z|Z, U_2) + I(S_1, Z; U_1|Z, U_2)\label{eqCommon4}\\
R_0+R_2&\geq H(Z|Z, U_1) + I(S_2, Z; U_2|Z, U_1)\label{eqCommon5}\\
R_1+R_2&\geq I(S_1, S_2, Z; U_1, U_2, Z|Z)\label{eqCommon6}\\
R_0+R_1+R_2&\geq H(Z)+I(S_1, S_2, Z; U_1, U_2, Z|Z)\label{eqCommon7}
\end{align}
and $\mathbb{E}[d_j(S_j, \hat{S}_j)]\leq D_j$ for $j=1,2$, for some distribution 
\begin{align}\label{eq:distr} 
&p(z,  s_1, s_2, u_1, u_2, \hat{s}_1, \hat{s}_2) =p(z, s_1, s_2) p(u_1|s_1, z) p(u_2|s_2, z) p(\hat{s}_1, \hat{s}_2| z, u_1, u_2).
\end{align}
Condition \eqref{eqCommon1} can be removed without loss of generality. 
We can write \eqref{eqCommon2} as, 
\begin{align} 
R_1&\geq I(S_1, Z; U_1|Z, U_2) \label{combine1} \\
&=H(U_1|Z, U_2)- H(U_1|S_1, Z, U_2) \label{combine12}\\
&=H(U_1|Z)- H(U_1|S_1, Z) \label{combine13}\\
&= I(S_1;U_1|Z) \label{combine14}
\end{align} 
where \eqref{combine13} is from $U_1-S_1Z-U_2$ and $U_1-Z-U_2$ since
\begin{align}
p(u_1, u_2|z) &= \sum_{s_1, s_2}  \!p(u_1| s_1, z) p(u_2| s_2, z) p(s_1| z)p(s_2| z) \label{markovcond2}\\
&= \sum_{s_1, s_2}  p(u_1, s_1|z) p(u_2, s_2|z) \label{markovcond3}\\
&=  p(u_1|z) p(u_2|z) \label{markovpu1u2}
\end{align} 
where \eqref{markovcond2} is from $U_1-S_1Z-S_2U_2$ and $U_2-S_2Z-S_1$ as well as $S_1-Z-S_2$.

Following the steps in \eqref{combine1}-\eqref{combine14}, we can write \eqref{eqCommon4} as
\begin{equation} \label{combine1old}
R_0+R_1\geq I(S_1;U_1|Z),
\end{equation}
which, comparing with \eqref{combine14}, indicates that \eqref{eqCommon4} can be removed without loss of generality. 

Following similar steps, we can write \eqref{eqCommon3} and  \eqref{eqCommon5} as
\begin{align} 
R_2&\geq I(S_2; U_2|Z)  \label{combine2old} \\
R_0+R_2&\geq I(S_2; U_2|Z)  \label{combine2}
\end{align}
respectively, which  show that condition \eqref{eqCommon5} can also be removed. 
For \eqref{eqCommon6}-\eqref{eqCommon7}, we find that
\begin{align}
I(S_1, S_2, Z; U_1, U_2, Z|Z) 
&=I(S_1, S_2; U_1, U_2|Z) \label{simplified}\\
&=H(U_1|Z)+H(U_2|Z, U_1)\!-\! H(U_1|Z, S_1)- H(U_2|Z, S_2) \label{eq:decompn} \\
&=H(U_1|Z)\!+\!H(U_2|Z)\!-\! H(U_1|Z, S_1)\!-\! H(U_2|Z, S_2) \label{eq:decompn2}\\
&=I(S_1;U_1|Z)+I(S_2;U_2|Z) \label{eq:decompn3}
\end{align}
where \eqref{eq:decompn} holds as $U_1-ZS_1-S_2$ and $U_2-ZS_2-S_1U_1$; and  \eqref{eq:decompn2} follows from $U_1-Z-U_2$ shown in \eqref{markovpu1u2}.

Combining   \eqref{combine14}, \eqref{combine1old}, \eqref{combine2old}, and \eqref{combine2} with \eqref{eq:decompn3}, we restate \eqref{eqCommon1}-\eqref{eq:distr} as follows.
A distortion pair $(D_1, D_2)$ is achievable for the rate triplet $(R_0, R_1, R_2)$ if
\begin{align}
R_1&\geq I(S_1;U_1|Z) \label{rr1}\\
R_2&\geq I(S_2;U_2|Z) \label{rr2}\\
R_1+R_2&\geq I(S_1;U_1|Z)+I(S_2;U_2|Z) \label{rr3}\\
R_0+R_1+R_2&\geq H(Z) \!+ \!I(S_1;U_1|Z)+I(S_2;U_2|Z) \label{rr4}
\end{align}
and $\mathbb{E}[d_j(S_j, \hat{S}_j)]\leq D_j$ for $j=1,2$, for some distribution 
\begin{equation}
p(z, s_1, s_2) p(u_1|s_1, z) p(u_2|s_2, z) p(\hat{s}_1, \hat{s}_2| z, u_1, u_2).
\end{equation}

We next show that one can set $\hat{S}_j=f_j(Z, U_1, U_2)$ for $j=$~$1,2$ without loss of optimality. 
To do so, we write
\begin{align}
\mathbb{E}[d_1(S_1, \hat{S}_1)]&=\sum_{s_1, \hat{s}_1} p(s_1, \hat{s}_1) d_1(s_1, \hat{s}_1)\\
&=\sum_{s_1, \hat{s}_1, \hat{s}_2, z, u_1, u_2} \hspace{-0.4cm}p( \hat{s}_1, \hat{s}_2|z, u_1, u_2, s_1) p(z, u_1, u_2, s_1) d_1(s_1, \hat{s}_1)
\\
&=\sum_{s_1, \hat{s}_1, \hat{s}_2, z, u_1, u_2}  p( \hat{s}_1, \hat{s}_2|z, u_1, u_2) p(z, u_1, u_2, s_1) d_1(s_1, \hat{s}_1)
\\
&=\sum_{z, u_1, u_2} \sum_{\hat{s}_1}  \sum_{s_1}   p( \hat{s}_1|z, u_1, u_2) p(z, u_1, u_2, s_1) d_1(s_1, \hat{s}_1) \\
&\geq \sum_{z, u_1, u_2, s_1} p(z, u_1, u_2, s_1) d_1(s_1, f_1(z, u_1, u_2)) \label{eq:f}\\
&=\mathbb{E}[d_1(S_1, f_1(Z, U_1, U_2))] \label{eqson93}
\end{align}
where we define a function $f_1: \mathcal{Z}\times \mathcal{U}_1\times \mathcal{U}_2\rightarrow \hat{\mathcal{S}}_1$ in \eqref{eq:f} such that,
\begin{equation}
f_1(z, u_1, u_2)=\arg \min_{\hat{s}_1} \sum_{s_1}  p(z, u_1, u_2, s_1) d_1(s_1, \hat{s}_1)
\end{equation}
and set $p( \hat{s}_1|z, u_1, u_2)=1$ for  $\hat{s}_1=f_1(z, u_1, u_2)$ and $p( \hat{s}_1|z, u_1, u_2)=0$ otherwise. 

A similar argument follows for $S_2$ by defining a function $f_2: \mathcal{Z}\times \mathcal{U}_1\times \mathcal{U}_2\rightarrow \hat{\mathcal{S}}_2$ leading to
\begin{equation} \label{functionfors2}
\mathbb{E}[d_2(S_2, \hat{S}_2)]\geq \mathbb{E}[d_2(S_2, f_2(Z, U_1, U_2))].
\end{equation}
Therefore, we can set $\hat{S}_j=f_j(Z, U_1, U_2)$ for $j=1,2$.

We next show for $j=1,2$ that whenever there exists a function $f_j(Z, U_1, U_2)$  such that 
\begin{equation}
\mathbb{E}[d_j(S_j, f_j(Z, U_1, U_2))]\leq D_j, 
\end{equation}
then there exists a function $g_j(Z, U_j)$ such that 
\begin{equation}
\mathbb{E}[d_j(S_j, g_j(Z, U_j))]\leq \!\mathbb{E}[d_j(S_j, f_j(Z, U_1, U_2))]\!\leq\! D_j. 
\end{equation}

We show this result along the lines of \cite{gastpar2004wyner}. 
Consider a function $f_1(Z, U_1, U_2)$ such that $\mathbb{E}[d_1(S_1, f_1(Z, U_1, U_2))]\leq D_1$.
From the law of iterated expectations, 
\begin{align}
\mathbb{E}[d_1(S_1, f_1(Z, U_1, U_2))] 
&= \mathbb{E}_{S_2, U_2, Z} [\mathbb{E}_{S_1, U_1|S_2, U_2, Z} [d_1(S_1, f_1( Z, U_1, U_2))]] \\
&=  \mathbb{E}_{S_2, U_2, Z} [\mathbb{E}_{S_1, U_1|Z} [d_1(S_1, f_1(Z, U_1, U_2))]]  \label{eq:LOIE}
\end{align}
\eqref{eq:LOIE} holds due to $U_1S_1-Z-U_2S_2$, see \eqref{markovcond2}-\eqref{markovcond3}. Define $
\phi: \mathcal{Z}\rightarrow \mathcal{U}_2$ such that 
\begin{equation}
\phi(z)\triangleq \arg \min_{u_2} \mathbb{E}_{S_1, U_1|Z=z} [d_1(S_1, f_1(z, U_1, u_2))].
\end{equation}
Then for each $Z=z$,
\begin{equation}
 \mathbb{E}_{S_2, U_2| Z=z} [\mathbb{E}_{S_1, U_1|Z=z} [d_1(S_1, f_1(z, U_1, U_2))]] \geq  \mathbb{E}_{S_1, U_1|Z=z} [d_1(S_1, f_1(z, U_1, \phi(z)))],
\end{equation}
and hence,
\begin{align}
\mathbb{E}[d_1(S_1, f_1(Z, U_1, U_2))] 
& = \mathbb{E}_{Z} [\mathbb{E}_{S_2, U_2| Z=z} [\mathbb{E}_{S_1, U_1|Z=z} [d_1(S_1, f_1(z, U_1, U_2))]]] \\
&\geq \mathbb{E}_{Z}   [\mathbb{E}_{S_1, U_1|Z=z} [d_1(S_1, f_1(z, U_1, \phi(z)))]] \\
&= \mathbb{E}_{S_1, U_1,Z} [d_1(S_1, f_1(Z, U_1, \phi(Z)))] \\
&= \mathbb{E}[d_1(S_1, g_1(Z, U_1))] \label{eqf1}
\end{align}
where $g_1(Z, U_1)=f_1(Z, U_1, \phi(Z))$. 

Following similar steps, for any $f_2(Z, U_1, U_2)$ that achieves $\mathbb{E}[d_2(S_2, f_2(Z, U_1, U_2))]\leq D_2$ we can find a function $g_2(Z, U_2)$ such that 
\begin{equation} \label{eqdist2}
\mathbb{E}[d_2(S_2, f_2(Z, U_1, U_2))]  \geq \mathbb{E}[d_2(S_2, g_2(Z, U_2))]. 
\end{equation}

Combining \eqref{eqson93}, \eqref{functionfors2}, \eqref{eqf1}, \eqref{eqdist2} with \eqref{func1} and \eqref{func2},  we can state the rate region in \eqref{rr1}-\eqref{rr4} as follows. 
A distortion pair $(D_1, D_2)$ is achievable for the rate triplet $(R_0, R_1, R_2)$ if
\begin{align}
R_1&\geq  R_{S_1|Z} (D_1) \label{rrn1}\\
R_2&\geq R_{S_2|Z} (D_2) \label{rrn2}\\
R_1+R_2&\geq R_{S_1|Z} (D_1)+R_{S_2|Z} (D_2)  \label{rrn3}\\
R_0+R_1+R_2&\geq H(Z) +  R_{S_1|Z} (D_1)+ R_{S_2|Z} (D_2)  \label{rrn4}
\end{align}
since for any $p(s_j, u_j, z)=p(u_j|s_j, z) p(s_j|z)p(z)$ and $g_j(z, u_j)$ with $\mathbb{E}[d_j(S_j, g_j(Z, U_j))]\leq D_j$,   
\begin{equation}
I(S_j;U_j|Z) \geq R_{S_j|Z} (D_j), \quad j=1,2, 
\end{equation}
where $R_{S_j|Z} (D_j)$ is defined in \eqref{func2}. This completes the source coding part.

Our channel coding is based on coding for a MAC with a common message \cite{slepian1973coding}, for which any triplet of rates $(R_0, R_1, R_2)$ is achievable if
\begin{align}
R_1 &\leq I(X_1;Y|X_2, W) \label{common1MAC}\\
R_2 &\leq I(X_2;Y|X_1, W)\label{common2MAC}\\
R_1+R_2&\leq I(X_1, X_2;Y| W)\label{common3MAC}\\
R_0+R_1+R_2 &\leq I(X_1, X_2;Y) \label{common4MAC}
\end{align}
for some $p(x_1, x_2, y, w)=p(y|x_1, x_2) p(x_1|w) p(x_2|w)p(w)$.

\subsection{Converse} 
Our proof is along the lines of \cite{shamai1998systematic} and  \cite{gunduz2007correlated}. 
Suppose there exist encoding functions $e_j: \mathcal{S}_j^n\times \mathcal{Z}^n\rightarrow \mathcal{X}_j^n$ for $j=1,2$, decoding functions $g_j: \mathcal{Y}^n\rightarrow \hat{\mathcal{S}}_j^n$ for $j=1,2$ and $g_0: \mathcal{Y}^n\rightarrow \hat{Z}^n$ 
such that $\frac{1}{n}\sum_{i=1}^n E[d_j(S_{ji}, \hat{\mathcal{S}}_{ji} )]\leq D_j+\epsilon$ for $j=1,2$ and $P(Z^n\neq \hat{Z}^n)\leq P_e$ where $\epsilon \rightarrow 0$, $P_e\rightarrow 0$ as $n\rightarrow \infty$.

Define $U_{ji}=(Y^n, S_j^{i-1}, Z_{i}^c)$ for $j=1,2$ where $Z_{i}^c=$~$(Z_{1}, \ldots, Z_{(i-1)}, Z_{(i+1)},  \ldots, Z_{n})$. Then, 
\begin{align}
\frac{1}{n} I(X_1^n; Y^n|X_2^n, Z^n)
&= \frac{1}{n} (H(Y^n|X_2^n, Z^n)- H(Y^n|X_1^n, X_2^n, Z^n, S_1^n))  \label{eq888n2}\\
&\geq \frac{1}{n} (H(Y^n|X_2^n, Z^n)- H(Y^n|X_2^n, Z^n, S_1^n))  \label{eq888n2n1}\\
&= \frac{1}{n} I(S_1^n; Y^n, X_2^n|Z^n) \label{eq888n5} \\
&\geq \frac{1}{n} I(S_1^n; Y^n|Z^n)  = \frac{1}{n}  \sum_{i=1}^n I(S_{1i}; Y^n|S_{1}^{i-1}, Z^n) \label{eq888n6} \\
&=\frac{1}{n} \sum_{i=1}^n I(S_{1i}; U_{1i}|Z_{i}) \label{eq888n7} \\
&\geq \frac{1}{n} \sum_{i=1}^n R_{S_1|Z}(\mathcal{E} (S_{1i}|U_{1i}, Z_{i})) \label{eq888n8} \\
&\geq \frac{1}{n} \sum_{i=1}^n R_{S_1|Z}(\mathcal{E} (S_{1i}|Y^n))  \label{eq888n9}  \\
&\geq \frac{1}{n} \sum_{i=1}^n R_{S_1|Z}(\mathbb{E} [d_1(S_{1i}, \hat{S}_{1i} )] ) \label{eq888n10} \\
&\geq  R_{S_1|Z}(D_1+\epsilon) \label{eq888n11} 
\end{align}
\eqref{eq888n2} is from $Y^n-X_1^n X_2^n-Z^n S_1^n$, \eqref{eq888n2n1} holds since conditioning cannot increase entropy, and \eqref{eq888n5} is from $I(S_1^n;X_2^n|Z^n)=0$ since $S_1^n-Z^n-X_2^n$ as follows. 
\begin{align}
p(x_2^n, s_1^n|z^n)&=\sum_{s_2^n} p(x_2^n, s_2^n, s_1^n|z^n) \\
&=\sum_{s_2^n} p(x_2^n| s_2^n, z^n) p(s_2^n| z^n)  p(s_1^n| z^n)  \label{neweq2}\\
&=p(x_2^n| z^n)  p(s_1^n| z^n)  
\end{align}
where \eqref{neweq2} holds since $X_2^n-S_2^nZ^n -S_1^n$ and $S_1^n-Z^n-S_2^n$.
Equation \eqref{eq888n7} is from the definition of $U_{1i}$ and the memoryless property of the sources; \eqref{eq888n8}  is from \eqref{func1} and \eqref{func2}; \eqref{eq888n9} is from the fact that conditioning cannot increase \eqref{func1}; \eqref{eq888n10} follows as $\hat{S}_{1i}$ is a function of $Y^n$ and \eqref{eq888n11} as $R_{S_1|Z}(D_1)$ is convex and monotone in $D_1$. 

By defining a discrete random variable $\tilde{Q}$ uniformly distributed over $\{1,\ldots, n\}$ independent of everything else, we find that
\begin{align}
\frac{1}{n} I(X_1^n; Y^n|X_2^n, Z^n)   
&\leq\frac{1}{n} \sum_{i=1}^n (H(Y_i|X_{2i}, Z^n) - H(Y_i|X_{1i}, X_{2i}, Z^n) )  \label{lemma2Ix1yn1}  \\
&=\frac{1}{n} \sum_{i=1}^n  I(X_{1i};Y_i|X_{2i}, \tilde{Q}=i, Z^n)  \\
&=I(X_{1\tilde{Q}};Y_{\tilde{Q}}|X_{2\tilde{Q}}, \tilde{Q}, Z^n)  \\
&= I(X_1;Y|X_2, W) \label{lemma2Ix1yn5}
\end{align}
where we let $X_1=X_{1\tilde{Q}}$, $X_2=X_{2\tilde{Q}}$, $Y=Y_{\tilde{Q}}$ and $W=(\tilde{Q}, Z^n)$. 
Combining \eqref{lemma2Ix1yn5} with \eqref{eq888n2} and \eqref{eq888n11} leads to \eqref{common1}. 
We obtain \eqref{common2} by following similar steps. 
Next, we show that
\begin{align}
\frac{1}{n}I(X_1^n, X_2^n; Y^n|Z^n)
&= \frac{1}{n} (H(Y^n|Z^n)- H(Y^n|Z^n, X_1^n, X_2^n))\label{eq88802}\\
&= \frac{1}{n} (H(Y^n|Z^n)- H(Y^n|Z^n, X_1^n, X_2^n, S_1^n, S_2^n))\label{eq88802next}\\
&\geq  \frac{1}{n} (H(Y^n|Z^n)- H(Y^n|Z^n, S_1^n, S_2^n))\label{eq88802next2}\\
&= \frac{1}{n} ( I(S_1^n; Y^n|Z^n) + H(S_2^n|Z^n)- H(S_2^n| Y^n, S_1^n, Z^n))\label{eq8882} \\
&\geq \frac{1}{n} ( I(S_1^n; Y^n|Z^n) + H(S_2^n|Z^n) - H(S_2^n| Y^n, Z^n)) \label{eqson42}\\
&\geq R_{S_1|Z}(D_1+\epsilon) + R_{S_2|Z}(D_2+\epsilon)   \label{88822}
\end{align} 
where \eqref{eq88802next} is from $Y^n-X_1^n X_2^n-S_1^n S_2^n Z^n $, \eqref{eq88802next2} is from the fact that conditioning cannot increase entropy, \eqref{eq8882} is from $S_2^n-Z^n-S_1^n$, \eqref{eqson42} is from conditioning cannot increase entropy, \eqref{88822} is from following the steps \eqref{eq888n6}-\eqref{eq888n11} twice, where the role of $S_1^n$ and $S_2^n $ are changed for the second term. 
Moreover, we have
\begin{align}
\frac{1}{n}I(X_1^n, X_2^n; Y^n | Z^n) 
&\leq \frac{1}{n}  \sum_{i=1}^n(H(Y_i|Z^n)- H(Y_i|X_{1i}, X_{2i}, Z^n))  \label{eqlemma2sumrate1}  \\
&= \frac{1}{n}  \sum_{i=1}^nI(X_{1i}, X_{2i}; Y_i|\tilde{Q}=i, Z^n) \\
&\leq I(X_{1\tilde{Q}}, X_{2\tilde{Q}}; Y_{\tilde{Q}}| \tilde{Q}, Z^n)\\
&\leq I(X_{1}, X_{2}; Y| W) \label{eqlemma2sumrate4}
\end{align} 
where $X_1=X_{1\tilde{Q}}$, $X_2=X_{2\tilde{Q}}$, $Y=Y_{\tilde{Q}}$ and $W=(\tilde{Q}, Z^n)$.  
Combining \eqref{eqlemma2sumrate4} with \eqref{eq88802} and \eqref{88822} leads to \eqref{common3}.
We lastly show that
\begin{align}
\frac{1}{n}I(X_1^n, X_2^n; Y^n)
&\geq \frac{1}{n} I(S_1^n, S_2^n, Z^n; Y^n) \label{eq8880}\\
&= \frac{1}{n} (I(Z^n; Y^n) + I(S_1^n; Y^n|Z^n) + H(S_2^n|Z^n)- H(S_2^n| Y^n, S_1^n, Z^n))\label{eq888} \\
&\geq \frac{1}{n} (I(Z^n; Y^n) + I(S_1^n; Y^n|Z^n) + H(S_2^n|Z^n) - H(S_2^n| Y^n, Z^n)) \label{eqson4}\\
&\geq  \frac{1}{n} (H(Z^n)  + I(S_1^n; Y^n|Z^n) + I(S_2^n;Y^n | Z^n)-n\delta(P_e)) \label{8881}\\
&\geq H(Z)  \!+\! R_{S_1|Z}(D_1\!+\!\epsilon) \!+\! R_{S_2|Z}(D_2\!+\!\epsilon)  - \delta(P_e) \label{8882}
\end{align} 
where \eqref{eq8880} is from $Y^n-X_1^n X_2^n-S_1^n S_2^n Z^n $, \eqref{eq888} is from $S_2^n-Z^n-S_1^n$, \eqref{eqson4} is from the fact that conditioning cannot increase entropy, 
\eqref{8881} is from Fano's inequality combined with the data processing inequality, i.e.,
\begin{equation}
H(Z^n|Y^n) \leq H(Z^n|\hat{Z}^n) \leq n\delta(P_e) \label{eq888last}
\end{equation}
where $\delta(P_e)\rightarrow 0$ as $P_e\rightarrow 0$ \cite{cover2012elements}. Equation \eqref{8882} is from the memoryless property of $Z^n$ and from following \eqref{eq888n6}-\eqref{eq888n11} twice, the second one is when the role of $S_1^n$ is replaced with $S_2^n $. 

Lastly, using random variable $\tilde{Q}$ that has been defined uniformly over $\{1, \ldots, n\}$ and independent of everything else, we derive the following. 
\begin{align}
\frac{1}{n}I(X_1^n, X_2^n; Y^n)  
&\leq \frac{1}{n}  \sum_{i=1}^n(H(Y_i)\!-\! H(Y_i|X_{1i}, X_{2i})) \label{eq888beginnew} \\
&= \frac{1}{n}  \sum_{i=1}^nI(X_{1i}, X_{2i}; Y_i|\tilde{Q}=i) \\
&\leq I(X_{1\tilde{Q}}, X_{2\tilde{Q}}; Y_{\tilde{Q}}| \tilde{Q})\\
&= I(X_{1}, X_{2}; Y| \tilde{Q}) \label{eq8871}\\
&\leq  H(Y) - H(Y|X_{1}, X_{2}) \\
&= I(X_{1}, X_{2}; Y) \label{eqcommon4last}
\end{align} 
where $X_1=X_{1\tilde{Q}}$, $X_2=X_{2\tilde{Q}}$, $Y=Y_{\tilde{Q}}$. 
Combining \eqref{eq8880}, \eqref{8882}, \eqref{eq888beginnew}, and \eqref{eqcommon4last} leads to \eqref{common4}.

In order to complete our proof, we demonstrate that $p(x_1, x_2|w)=p(x_1|w) p(x_2|w)$ for  $w=(i, z^n)$. To this end, we show that
\begin{align}
P(X_1=x_1, X_2=x_2| W=w) 
&=P(X_{1i}=x_1, X_{2i}=x_2| \tilde{Q}=i, Z^n=z^n) \\
&=P(X_{1i}\!=\!x_1| \tilde{Q}\!=\!i, Z^n\!=\!z^n)   P(X_{2i}\!=\!x_2| \tilde{Q}\!=\!i, Z^n\!=\!z^n)  \label{distx1x2q2}\\ 
&=P(X_{1}=x_1| W=w)  P(X_{2}=x_2| W=w)  
\end{align}
where \eqref{distx1x2q2} holds since $X_{1i}-Z^n-X_{2i}$ for $i=1,\ldots, n$ as follows.
\vspace{-0.1cm}\begin{align}
p(x_1^n, x_2^n|z^n)
&=\sum_{s_1^n, s_2^n} p(x_1^n, x_2^n, s_1^n, s_2^n|z^n)  \label{eq:indept00} \\
&=\sum_{s_1^n, s_2^n} p(x_1^n| s_1^n, z^n) p(x_2^n| s_2^n, z^n) p(s_1^n| z^n)  p(s_2^n| z^n) \label{eqmarkovdist}\\
&= p(x_1^n|  z^n) p(x_2^n| z^n)  \label{eq:indept}
\end{align} 
where \eqref{eqmarkovdist} is from $X_1^n-S_1^n Z^n -S_2^nX_2^n$ and $X_2^n-S_2^n Z^n -S_1^n$ as well as $S_1^n - Z^n -S_2^n$.  
From \eqref{eq:indept}, we observe that  $X_{1}^n-Z^n-X_{2}^n$, which implies $X_{1i}-Z^n-X_{2i}$.

\section{Proof of Theorem~\ref{lemma2}}\label{appendixA}

\subsection{Achievability}
The source coding part is based on lossy source coding at the two encoders conditioned on the side information $Z$ shared between the encoder and decoder  \cite{gray1972conditional}, after which the conditional rate distortion functions given in \eqref{func2} can be achieved for $S_1$ and $S_2$, respectively. Channel coding part is based on coding for a classical MAC  with independent channel inputs \cite{cover2012elements}.

\subsection{Converse}
Suppose there exist encoding functions $e_j: \mathcal{S}_j^n\times \mathcal{Z}^n\rightarrow \mathcal{X}_j^n$, $j=1,2$, and decoding functions $g_j: \mathcal{Y}^n\times \mathcal{Z}^n\rightarrow \hat{\mathcal{S}}_j^n$ such that $\frac{1}{n}\sum_{i=1}^n E[d_j(S_{ji}, \hat{\mathcal{S}}_{ji} )]\leq D_j+\epsilon$, where $\epsilon \rightarrow 0$ as $n\rightarrow \infty$.  
Then, 
\begin{align}
\frac{1}{n} I(X_1^n;Y^n|X_2^n, Z^n) 
&\geq \frac{1}{n} I(S_1^n; Y^n|X_2^n, Z^n) \label{eq2}\\
&= \frac{1}{n} I(S_1^n; Y^n, X_2^n| Z^n) \label{eq3}\\
&\geq \frac{1}{n} I(S_1^n; Y^n| Z^n) \label{eq4}\\
&= \frac{1}{n} H(S_1^n| Z^n)- H(S_1^n| Y^n, Z^n, \hat{S}_1^n) \label{eq5new}\\
&\geq \frac{1}{n} H(S_1^n| Z^n)- H(S_1^n| Z^n, \hat{S}_1^n) \label{eq6new}\\
&\geq \frac{1}{n} \sum_{i=1}^n( H(S_{1i}| Z_i)- H(S_{1i}| Z_i, \hat{S}_{1i})) \label{eq8new}\\
&\geq \frac{1}{n} \sum_{i=1}^n I(S_{1i};\hat{S}_{1i}| Z_i) \label{eq9new}\\
&\geq \frac{1}{n} \sum_{i=1}^n R_{S_1|Z} (E[d_1(S_{1i}, \hat{S}_{1i})]) \label{eq9}\\
&\geq  R_{S_1|Z} (D_1+\epsilon) \label{eq10}
\end{align}
\eqref{eq2} is from $Y^n-X_1^nX_2^n-S_1^nZ^n$ and conditioning cannot increase entropy, and \eqref{eq3} is from $X_2^n-Z^n-S_1^n$ which holds since
\begin{equation}
p(x_2^n, s_1^n|z^n)\!=\!\sum_{s_2^n} p(x_2^n, s_1^n, s_2^n|z^n) 
\!=\!\sum_{s_2^n} p(x_2^n| s_2^n,z^n) p(s_1^n|z^n) p(s_2^n|z^n)\!=\!p(x_2^n|z^n) p(s_1^n|z^n) \! \label{eqstar}
\end{equation}
from $X_2^n-S_2^nZ^n-S_1^n$ and  $S_1^n-Z^n-S_2^n$. Equation 
\eqref{eq4} is due to the nonnegativity of mutual information; \eqref{eq5new} follows from $\hat{S}_1^n=g_1(Y^n, Z^n)$; \eqref{eq6new} holds since conditioning cannot increase entropy; 
\eqref{eq8new} is from the memoryless property of the sources and the side information as well as the chain rule and the fact that conditioning cannot increase entropy; \eqref{eq10} holds as $R_{S_1|Z}(D_1)$ is convex and monotone in $D_1$.

By defining a discrete uniform random variable $\tilde{Q}$ over $\{1,\ldots, n\}$ independent of everything else, and following steps \eqref{lemma2Ix1yn1}-\eqref{lemma2Ix1yn5} by $W=(\tilde{Q}, Z^n)$ replaced with $Q=(\tilde{Q}, Z^n)$, we find that 
\begin{equation}
\frac{1}{n} I(X_1^n; Y^n|X_2^n, Z^n)\leq I(X_1;Y|X_2, Q) \label{eqbeforelast}
\end{equation} 
where $X_1=X_{1\tilde{Q}}$, $X_2=X_{2\tilde{Q}}$, $Y=Y_{\tilde{Q}}$.
Combining \eqref{eq2}, \eqref{eq10}, and \eqref{eqbeforelast} yields \eqref{lemm2cond1}. Following similar steps we obtain  \eqref{lemm2cond2},
\begin{equation}\label{eqsource2}
R_{S_2|Z} (D_2+\epsilon) \leq I(X_2;Y|X_1, Q).
\end{equation}
Lastly, we have
\begin{align}
\frac{1}{n} I(X_1^n, X_2^n;Y^n|Z^n)
&=\frac{1}{n} I (X_1^n;Y^n|X_2^n, Z^n)+ \frac{1}{n}I(X_2^n;Y^n|Z^n) \label{eqsumrate2}\\
&\geq  R_{S_1|Z} (D_1+\epsilon) +\frac{1}{n} I(S_2^n;Y^n|Z^n)\label{eqsumrate3}\\
&\geq R_{S_1|Z} (D_1+\epsilon) +  R_{S_2|Z} (D_2+\epsilon)\label{eqsumrate21}
\end{align}
where the first term in \eqref{eqsumrate3} is from \eqref{eq2}-\eqref{eq10}, and \eqref{eqsumrate21} follows similarly to \eqref{eq4}-\eqref{eq10}. To obtain the second term in \eqref{eqsumrate3}, we first show that $Y^n-Z^nX_2^n-S_2^n$:
\begin{align}
p(y^n, s_2^n|z^n, x_2^n) 
&= p(s_2^n|z^n, x_2^n) p(y^n|s_2^n, z^n, x_2^n)  \label{Markov1}\\
&=p(s_2^n|z^n, x_2^n) \sum_{s_1^n, x_1^n} p(y^n|x_1^n, x_2^n) p(x_1^n|s_1^n, z^n) p(s_1^n|z^n)\label{Markov4} \\
& =p(s_2^n|z^n, x_2^n) \sum_{ x_1^n} p(y^n|x_1^n, x_2^n) p(x_1^n| z^n) \label{Markov5} 
\end{align}
\eqref{Markov4} is from $Y^n-X_1^nX_2^n- S_1^n S_2^n Z^n$ and $X_1^n-S_1^nZ^n-S_2^nX_2^n$ as well as $S_1^n-Z^n-S_2^nX_2^n$, which holds since
\begin{equation}
p(s_1^n, s_2^n, x_2^n|z^n) =p(x_2^n| s_2^n, z^n) p( s_2^n| z^n)p( s_1^n| z^n) =p(x_2^n, s_2^n| z^n) p( s_1^n| z^n), 
\end{equation}
due to $X_2^n-S_2^nZ^n-S_1^n$ and $S_1^n-Z^n-S_2^n$. 
Note that 
\begin{align}
p(y^n|z^n, x_2^n)
&= \sum_{s_1^n, x_1^n} p(y^n|x_1^n, x_2^n) p(x_1^n| s_1^n,z^n)p(s_1^n|z^n)= \sum_{x_1^n} p(y^n|x_1^n, x_2^n) p(x_1^n|z^n), \label{eq:markovbeforelast}
\end{align}
as $X_1^n-S_1^nZ^n-X_2^n$ and $S_1^n-Z^n-X_2^n$ holds since  $S_1^n-Z^n-S_2^nX_2^n$.
From \eqref{eq:markovbeforelast} and \eqref{Markov5}, 
\vspace{-0.2cm}\begin{equation}
p(y^n, s_2^n|z^n, x_2^n)=p(s_2^n|z^n, x_2^n)p(y^n|z^n, x_2^n),
\vspace{-0.2cm}\end{equation} 
and hence, $Y^n-Z^nX_2^n-S_2^n$. 
Then, we use the following in \eqref{eqsumrate2},
\begin{align}
I(X_2^n;Y^n|Z^n)&=H(Y^n|Z^n)-H(Y^n|X_2^n, Z^n, S_2^n) \label{eqMarkovcont2} \\
&\geq H(Y^n|Z^n)-H(Y^n|Z^n, S_2^n)  = I(S_2^n;Y^n|Z^n), \label{eqMarkovcont3}
\end{align}
where \eqref{eqMarkovcont2} is from $Y^n-Z^nX_2^n-S_2^n$, and \eqref{eqMarkovcont3} holds since conditioning cannot increase entropy, which leads to the second term in \eqref{eqsumrate3}.

Then, by replacing $W=(\tilde{Q}, Z^n)$ with $Q=(\tilde{Q}, Z^n)$ in  \eqref{eqlemma2sumrate1}-\eqref{eqlemma2sumrate4}, we can show by following the same steps that, 
\begin{equation}
\frac{1}{n}I(X_1^n, X_2^n; Y^n | Z^n) \leq I(X_{1}, X_{2}; Y| Q)  \label{eqlemma2sumrateQ}
\end{equation} 

Combining \eqref{eqsumrate2}, \eqref{eqsumrate21} and  \eqref{eqlemma2sumrateQ} recovers \eqref{lemm2cond3}. 
Lastly, we show  $p(x_1, x_2|q)=p(x_1|q) p(x_2|q)$ along the lines of  \cite{gunduz2009source}. For $q=(i, z^n)$, 
\begin{align}
P(X_1=x_1, X_2=x_2| Q=q)
&=P(X_{1i}=x_1, X_{2i}=x_2| \tilde{Q}=i, Z^n=z^n) \label{distx1x2q0}\\
&=P(X_{1i}\!=\!x_1| \tilde{Q}\!=\!i, Z^n\!=\!z^n)   P(X_{2i}\!=\!x_2| \tilde{Q}\!=\!i, Z^n\!=\!z^n)  \label{distx1x2qn}\\ 
&=P(X_{1}=x_1| Q=q)  P(X_{2}=x_2| Q=q)  \label{distx1x2q1}
\end{align}
where \eqref{distx1x2qn} holds since $X_{1i}-Z^n-X_{2i}$ for $i=1,\ldots, n$.


\section{Proof of Proposition~\ref{propbounds1}}\label{appendixD}

Let $\rho = 0.5$ and $P=2$.   
Partition the set of all distortion pairs $(D_1, D_2)$ for $0\leq D_1, D_2\leq 1$ as in Fig.~\ref{regions2}.  
First, consider $D_1=0.145$. For this case, one can observe that \eqref{LT} is satisfied with equality when $D_2 = 0.7476$, by noting that $(D_1, D_2)\in \mathcal{D}$ for $(D_1, D_2) = (0.145,  0.7476)$ and solving the resulting equation. Accordingly, for all distortion pairs $(0.145, D_2)$ with $0.7476\leq D_2\leq 1$, the necessary condition from \eqref{LT} is satisfied. 

Consider now the necessary conditions from Corollary~\ref{cor6} given in \eqref{eq3GGY1}-\eqref{eq3GGY2} along with the distortion pair $(D_1, D_2) = (0.145, 1)$,
\begin{align} 
 \frac{1}{2}\log \left (\frac{1-\rho}{D_1}\right )  &\leq \frac{1}{2} \log ( 1 + \beta_1 P + \beta_2 P ) \label{eq3GGY12}\\
 \frac{1}{2}\log \left (\frac{1}{D_1}\right)  &\leq \frac{1}{2} \log ( 1 + 2P + 2P \sqrt{(1-\beta_1)(1-\beta_2)} ), \label{eq3GGY22}
\end{align}
which follows from $R_{S_2|Z} (D_2) = 0$ when $D_2=1\geq 1-\rho$. By rearranging the terms in \eqref{eq3GGY12}, 
\begin{equation}\label{rearr}
\beta_1\geq \frac{\left(\frac{1-\rho}{D_1}\right)-1}{P} -\beta_2
\end{equation}
from which, by combining with \eqref{eq3GGY22}, we have the condition
\begin{equation}\label{cond12345}
\left(1-\left(\frac{\frac{1-\rho}{D_1}-1}{P} -\beta_2\right)\right)(1-\beta_2) \geq(1-\beta_1)  (1-\beta_2)\geq \left(\frac{ \frac{1}{D_1}-1-2P}{2P}\right)^2, 
\end{equation}
leading to
\begin{equation}\label{parabola}
-\beta_2^2 + \frac{ \frac{1-\rho}{D_1}-1}{P} \beta_2 + 1- \frac{ \frac{1-\rho}{D_1}-1}{P} - \Bigg(\frac{\frac{1}{D_1}-1-2P}{2P}\Bigg)^2 \geq 0.
\end{equation}
By substituting $D_1=0.145$, $\rho =0.5$, and $P=2$, we find that the left hand side (LHS) of \eqref{parabola} is a concave quadratic polynomial whose maximum value is $-0.0743$, attained when $\beta_2 = \frac{ \frac{1-\rho}{D_1}-1}{2P}= 0.6121$. Hence, \eqref{parabola} is not satisfied for any $0\leq\beta_2\leq 1$, and no distortion pair $(0.145, D_2)$ for which $0\leq D_2\leq 1$ is achievable according to conditions \eqref{eq3GGY1}-\eqref{eq3GGY2}.

Lastly, consider the necessary conditions \eqref{eq3LW1}-\eqref{eq3LW2}. Consider the distortion pair $(D_1,D_2) = (0.145, 0.7476)$. Observe that $(0.145, 0.7476)\in\mathcal{D}$, as a result, \eqref{eq3LW1}-\eqref{eq3LW2} can be written as
\begin{align} 
\frac{1}{2}\log  \frac{(1-\rho)^2}{D_1 D_2 - \left(\rho - \sqrt{(1-D_1)(1-D_2)}\right)^2}  &\leq \frac{1}{2} \log ( 1 + \beta_1 P + \beta_2 P ) \label{eq3GGY12LW}\\
\frac{1}{2}\log  \frac{1-\rho^2}{D_1 D_2 - \left(\rho - \sqrt{(1-D_1)(1-D_2)}\right)^2} &\leq \frac{1}{2} \log ( 1 + 2P + 2P \sqrt{(1-\beta_1)(1-\beta_2)} ). \label{eq3GGY22LW}
\end{align}
Define $\alpha \triangleq \frac{(1-\rho)^2}{D_1 D_2 - \left(\rho - \sqrt{(1-D_1)(1-D_2)}\right)^2}$, and set $\beta_1 = \frac{\alpha-1}{P}-\beta_2$, which satisfies \eqref{eq3GGY12LW}. Then, \eqref{eq3GGY22LW} can be expressed as
\begin{equation}\label{pol3}
-\beta_2^2 + \frac{\alpha-1}{P}\beta_2 + 1- \frac{\alpha-1}{P} -\theta \geq 0,
\end{equation}
where $\theta \triangleq \left(\frac{(1-\rho^2)/(2P)}{D_1 D_2 - \left(\rho - \sqrt{(1-D_1)(1-D_2)}\right)^2} -\frac{1}{2P}-1\right)^2$. The LHS of \eqref{pol3} is a concave polynomial whose maximum value is $0.1945$, attained when $\beta_2 = \frac{\alpha-1}{2P} = 0.3333$, which satisfies \eqref{pol3}. The corresponding $\beta_1$ can  be computed from $\beta_1 = \frac{\alpha-1}{P}-\beta_2=\frac{\alpha-1}{2P} = 0.3333$. Hence, for all distortion pairs  $ (0.145, D_2)$ with  $0.7476\leq D_2\leq 1$, necessary conditions from  \eqref{eq3LW1}-\eqref{eq3LW2} are satisfied. 
Accordingly, we conclude that there exist distortion pairs $(D_1, D_2)$ in regions $\mathcal{G}$ and $\mathcal{D}$ that satisfy the conditions \eqref{LT} and  \eqref{eq3LW1}-\eqref{eq3LW2} but not \eqref{eq3GGY1}-\eqref{eq3GGY2}.

Next, consider $D_1=0.16$.  For this case, \eqref{LT} holds with equality when $D_2 =  0.6696$, by noting that $(0.16, D_2)\in \mathcal{B}$ for $(D_1, D_2) = (0.16,  0.702)$ and solving the resulting equation. The necessary condition from \eqref{LT} is then satisfied for all distortion pairs $(0.16, D_2)$ such that $0.6696\leq D_2\leq 1$. 

Consider next the conditions from \eqref{eq3GGY1}-\eqref{eq3GGY2} for $(D_1, D_2)\in \mathcal{B}$,
\begin{align} 
 \frac{1}{2}\log \left (\frac{1-\rho}{D_1}\right )  &\leq \frac{1}{2} \log ( 1 + \beta_1 P + \beta_2 P ) \\
 \frac{1}{2}\log \left (\frac{1-\rho^2}{D_1D_2}\right)  &\leq \frac{1}{2} \log ( 1 + 2P + 2P \sqrt{(1-\beta_1)(1-\beta_2)} ),
\end{align}
from which, as in \eqref{cond12345}, we can obtain the condition
\begin{equation}\label{cond123456}
\left(1-\left(\frac{\frac{1-\rho}{D_1}-1}{P} -\beta_2\right)\right)(1-\beta_2) \geq(1-\beta_1)  (1-\beta_2)\geq \left(\frac{ \frac{1-\rho^2}{D_1D_2}-1-2P}{2P}\right)^2, 
\end{equation}
and
\begin{equation}\label{parabola2}
-\beta_2^2 + \frac{ \frac{1-\rho}{D_1}-1}{P} \beta_2 + 1- \frac{ \frac{1-\rho}{D_1}-1}{P} - \Bigg(\frac{\frac{1-\rho^2}{D_1D_2}-1-2P}{2P}\Bigg)^2 \geq 0.
\end{equation}
By substituting $D_1=0.16$, $\rho =0.5$, and $P=2$, we observe that the LHS of \eqref{parabola2} is a concave quadratic polynomial whose maximum value occurs at $\beta_2 = 0.5312$. 
We note that whenever
$D_2 < 0.6818$, the LHS of \eqref{parabola2} is negative for all $0\leq \beta_2\leq 1$, hence the necessary conditions from Corollary~\ref{cor6} cannot be satisfied. 

Consider next conditions \eqref{eq3LW1}-\eqref{eq3LW2} for $(D_1,D_2) = (0.16, 0.6696)$. Since $(0.16, 0.6696)\!\in\!\mathcal{B}$, one can write \eqref{eq3LW1}-\eqref{eq3LW2} as
\begin{align} 
\frac{1}{2}\log  \frac{(1-\rho)^2}{D_1D_2}  &\leq \frac{1}{2} \log ( 1 + \beta_1 P + \beta_2 P ) \label{eq3GGY12LW2}\\
\frac{1}{2}\log  \left(\frac{1-\rho^2}{D_1D_2} \right) &\leq \frac{1}{2} \log ( 1 + 2P + 2P \sqrt{(1-\beta_1)(1-\beta_2)} ). \label{eq3GGY22LW2}
\end{align}
Define $\bar{\alpha} \triangleq  \frac{(1-\rho)^2}{D_1D_2}$. By letting $\beta_1 = \frac{\bar{\alpha}-1}{P}-\beta_2$, which satisfies \eqref{eq3GGY12LW2}, we restate \eqref{eq3GGY22LW2} as
\begin{equation}\label{pol32}
-\beta_2^2 + \frac{\bar{\alpha}-1}{P}\beta_2 + 1- \frac{\bar{\alpha}-1}{P} -\bar{\theta} \geq 0, 
\end{equation}
where $\bar{\theta} \triangleq \left(\frac{1-\rho^2}{2PD_1D_2}  -\frac{1}{2P}-1\right)^2$. The LHS of \eqref{pol32} is a concave polynomial with a maximum value of $0.1943$, attained when $\beta_2 = \frac{\bar{\alpha}-1}{2P} = 0.3334$, which satisfies \eqref{pol32}. The corresponding $\beta_1$ is computed from $\beta_1 = \frac{\bar{\alpha}-1}{P}-\beta_2=\frac{\bar{\alpha}-1}{2P} = 0.3334$. Therefore, for all distortion pairs  $(0.16, D_2)$ such that $0.6696\leq D_2\leq 1$, necessary conditions in \eqref{eq3LW1}-\eqref{eq3LW2} are satisfied. 
Since $(0.16, D_2)\in\mathcal{B}$ for all $0.6696  \leq D_2\leq 0.6818$, we conclude that there exist distortion pairs in Region $B$ that satisfy the necessary conditions from \eqref{LT} and from \eqref{eq3LW1}-\eqref{eq3LW2}, but not from Corollary~\ref{cor6}.

Lastly, consider the conditions from \eqref{nc1}-\eqref{nc2}. 
Note that $D_1\leq D_2$ in regions $\mathcal{B}$, $\mathcal{D}$, and $\mathcal{G}$, therefore \eqref{nc1}-\eqref{nc2} can be stated as,
\begin{align}
\frac{1}{(1-\hat{\rho})^2} \ln \left( \frac{1-\rho^2}{D_1}\right)&\leq P \label{nc3}\\
 (\ln2) R_{S_1 S_2}(D_1, D_2)&\leq  P(1+\hat{\rho})\label{nc4}
\end{align} 
for some $0\leq \hat{\rho}\leq |\rho|$.  
Note that, if
\begin{align}
R_{S_1 S_2}(D_1, D_2)  \leq \frac{P}{\ln2}, \label{nc5}
\end{align}
then, \eqref{nc4} is satisfied for any $\hat{\rho}$. 
For Region $\mathcal{B}$, we find from \eqref{nc5} that,
\begin{align}
\frac{1}{2} \log \left( \frac{1-\rho^2}{D_1(1-\rho)}\right)  \leq \frac{P}{\ln2} \label{nc6}
\end{align}
by letting $D_2=1-\rho$, which then leads to
\begin{align} \label{select1}
D_1\geq (1+\rho) 2^{-\frac{2P}{\ln2}}. 
\end{align}
If \eqref{nc5} is satisfied for some $(D_1, D_2)$, it will be satisfied for all $(D_1, D'_2)$ such that $D'_2\geq D_2$. 
Accordingly, if $1-\rho\geq D_1\geq (1+\rho) 2^{-\frac{2P}{\ln2}}$, then condition \eqref{nc4} is satisfied for all $D_2\geq 1-\rho$, irrespective of $\hat{\rho}$. 
Next, consider condition \eqref{nc3} and select $\hat{\rho}=0$, from which we have
\vspace{-0.2cm}\begin{equation}
P\geq \ln \left( \frac{1-\rho^2}{D_1}\right),
\end{equation}
or equally
\begin{align}\label{select2}
D_1 \geq (1-\rho^2)e^{-P}.
\end{align}

For $P=2$ and $\rho=0.5$, \eqref{select1} becomes $D_1\geq 0.0275$ and \eqref{select2} becomes $D_1\geq 0.1015$. Hence, both \eqref{nc1} and \eqref{nc2} are satisfied when $D_1=0.145$ and $D_1=0.16$. 

These examples demonstrate that there exist distortion pairs in regions  $\mathcal{B}$, $\mathcal{D}$, and $\mathcal{G}$, and from symmetry, in regions  $\mathcal{C}$, $\mathcal{F}$, and $\mathcal{I}$, for which the necessary conditions from  Corollary~\ref{cor6} is tighter than both \eqref{LT}, \eqref{nc1}-\eqref{nc2}, and \eqref{eq3LW1}-\eqref{eq3LW2}.

Lastly, we compare Corollary~\ref{cor6} with the conditions from \eqref{eq3LW1}-\eqref{eq3LW2} by investigating the LHS of both conditions for various regions in Fig.~\ref{regions2}, as the region defined by the RHS of both \eqref{eq3GGY1}-\eqref{eq3GGY2} and \eqref{eq3LW1}-\eqref{eq3LW2} is the same. 

For $(D_1, D_2)\in \mathcal{A}$, we observe from \eqref{point1} and \eqref{RateGausLHS} that,
\begin{equation}
R_{S_1 S_2} (D_1, D_2) \!-\! C_W(S_1, S_2) 
\!=\! \frac{1}{2} \log \left( \frac{1-\rho^2}{D_1 D_2}\right )  \!-\! 
\frac{1}{2} \log  \left( \frac{1\!+\!\rho}{1\!-\!\rho}\right ) \!=\! R_{S_1|Z} (D_1) \!+\! R_{S_2|Z} (D_2), 
\end{equation} 
hence, in this region, Corollary~\ref{cor6} and the \eqref{eq3LW1}-\eqref{eq3LW2} bound are equivalent. 

For $(D_1, D_2)\in \mathcal{B}$, we find from \eqref{point1} and \eqref{RateGausLHS} that,
\begin{align}
R_{S_1 S_2} (D_1, D_2)  - C_W(S_1, S_2) 
&=\frac{1}{2}  \log \left( \frac{1-\rho^2}{D_1 D_2}\right ) - \frac{1}{2} \log  \left( \frac{1+\rho}{1-\rho}\right ) \\
&\leq \frac{1}{2} \log \frac{1-\rho}{D_1} =R_{S_1|Z} (D_1) + R_{S_2|Z} (D_2),
\end{align} 
since $D_1\leq 1-\rho$ and $D_2\geq 1-\rho$ for $(D_1, D_2)\in \mathcal{B}$.
Hence, in this region, Corollary~\ref{cor6} is at least as tight as \eqref{eq3LW1}-\eqref{eq3LW2}. By swapping the roles of $D_1$ and $D_2$, we can extend the same argument to Region $\mathcal{C}$ as well.

For $(D_1, D_2)\in \mathcal{D}$, we have from \eqref{point1} and \eqref{RateGausLHS} that, 
\begin{equation} \label{eqcompare2}
R_{S_1|Z} (D_1) + R_{S_2|Z} (D_2)  = \frac{1}{2} \log \frac{1-\rho}{D_1}, 
\end{equation} 
whereas
\begin{align}
R_{S_1 S_2} (D_1, D_2)\!-\! C_W(S_1, S_2)&\!=\!\frac{1}{2} \max \Bigg\{ \! \log  \frac{1\!-\!\rho}{1\!+\!\rho}, \log  \frac{(1\!-\!\rho)^2}{D_1 D_2 \!-\! \left(\rho \!-\! \sqrt{(1\!-\!D_1)(1\!-\!D_2)}\right)^2}\!\Bigg \} \label{obs}\\
&\!=\!\frac{1}{2}  \log  \frac{(1-\rho)^2}{D_1 + D_2 - (1+\rho^2) +  2\rho \sqrt{(1-D_1)(1-D_2)}}, 
\label{general1}
\end{align}  
where the last equation follows from
\begin{align}
{(2-D_1-D_2)}^2 - 4\rho^2 (1-D_1)(1-D_2) = (1-\rho^2)(2-D_1-D_2)^2 + \rho^2(D_1-D_2)^2\geq 0 \label{2d1d2}
\end{align}
and therefore, 
\begin{equation}
 D_1 + D_2 - (1+\rho^2)+ 2\rho \sqrt{(1-D_1)(1-D_2)}\leq 1-\rho^2 . 
\end{equation}
Then, by comparing \eqref{general1} with \eqref{eqcompare2}, we find that,  Corollary~\ref{cor6} provides necessary conditions at least as tight as \eqref{eq3LW1}-\eqref{eq3LW2} if
\begin{equation} 
\rho \in \left\{\rho: \tau - \sqrt{D_2-1+\tau^2} \leq \rho \leq \tau + \sqrt{D_2-1+\tau^2}, \quad D_2+\tau^2\geq 1\right\},\end{equation}
where
\begin{equation}
\tau = \frac{D_1}{2} + \sqrt{(1-D_1)(1-D_2)}.
\end{equation}
By symmetry, for region $(D_1, D_2)\in \mathcal{F}$, 
Corollary~\ref{cor6} is at least as tight as \eqref{eq3LW1}-\eqref{eq3LW2} if 
\begin{equation} 
\rho \in \left\{\rho: \lambda - \sqrt{D_1-1+\lambda^2} \leq \rho \leq \lambda + \sqrt{D_1-1+\lambda^2}, \quad D_1+\tau^2\geq 1  \right\}, \end{equation}
where
\begin{equation}
\lambda = \frac{D_2}{2} + \sqrt{(1-D_1)(1-D_2)}.
\end{equation}

For $(D_1, D_2)\in \mathcal{G}$, we observe from \eqref{point1} and \eqref{RateGausLHS} that,
\begin{align}
R_{S_1 S_2} (D_1, D_2)- C_W(S_1, S_2)&=
\frac{1}{2} \log \left ( \frac{1}{D_1}\right ) - \frac{1}{2} \log  \left( \frac{1+\rho}{1-\rho}\right ) \\
&\leq \frac{1}{2} \log \frac{1-\rho}{D_1}=R_{S_1|Z} (D_1) + R_{S_2|Z} (D_2).
\end{align} 
Therefore, Corollary~\ref{cor6} is again at least as tight as \eqref{eq3LW1}-\eqref{eq3LW2}. 
It follows by symmetry that Corollary~\ref{cor6} is at least as tight as \eqref{eq3LW1}-\eqref{eq3LW2} in Region $\mathcal{I}$ as well.

For $(D_1, D_2)\in \mathcal{H}$, we have from \eqref{point1} and \eqref{RateGausLHS} that,
\begin{align}
R_{S_1 S_2} (D_1, D_2)- C_W(S_1, S_2)&=
\frac{1}{2} 
\log \left( \frac{1}{\min(D_1, D_2)}\right ) 
 - \frac{1}{2} \log  \left( \frac{1+\rho}{1-\rho}\right ) \\
&=\frac{1}{2} 
\log  \frac{1-\rho}{\min(D_1, D_2) (1+\rho)} \\
&\leq 0 =R_{S_1|Z} (D_1) + R_{S_2|Z} (D_2)  \label{triv}
\end{align} 
since $\min(D_1, D_2)\geq 1-\rho$ when $(D_1, D_2)\in H$. 
From \eqref{triv}, conditions \eqref{eq3GGY1} and \eqref{eq3LW1} are both trivially satisfied in this region, and therefore Corollary~\ref{cor6} and the conditions from \eqref{eq3LW1}-\eqref{eq3LW2} are equivalent. Same conclusion follows for Region $\mathcal{J}$.

For region $(D_1, D_2)\in \mathcal{E}$, we have from \eqref{point1} and \eqref{RateGausLHS} that, 
\vspace{-0.2cm}\begin{equation}\label{eqcompare}
R_{S_1|Z} (D_1) + R_{S_2|Z} (D_2)  = 0,
\vspace{-0.2cm}\end{equation} 
hence, condition \eqref{eq3GGY1} is trivially satisfied, whereas $R_{S_1 S_2} (D_1, D_2)- C_W(S_1, S_2)$ is as given in \eqref{obs} and \eqref{general1}.

If $D_1=D_2$, we have from \eqref{obs} and $D_1\geq 1-\rho$ that,
\begin{align}
R_{S_1 S_2} (D_1, D_2)- C_W(S_1, S_2)&=\frac{1}{2} \max \left \{  \log   \frac{1-\rho}{1+\rho}, \log  \frac{(1-\rho)^2}{D_1^2 - \left(\rho - (1-D_1)\right)^2} \right \} \\
& \leq 0=R_{S_1|Z} (D_1) + R_{S_2|Z} (D_2),
\end{align} 
and \eqref{eq3LW1} is also trivially satisfied. Hence, for all $D_1=D_2$ in Region $\mathcal{E}$, Corollary~\ref{cor6} and the conditions from \eqref{eq3LW1}-\eqref{eq3LW2} are equivalent. 

We next consider the case when $\rho\leq 0.5$ for $(D_1, D_2)\in \mathcal{E}$. 
Without loss of generality, we assume that $D_1 \geq D_2$. Noting that $D_2\geq 1-\rho$, we have
\vspace{-0.1cm}\begin{align}
D_1 + D_2 - (1+\rho^2) +  2\rho \sqrt{(1-D_1)(1-D_2)} 
&\geq D_1 + D_2 - (1+\rho^2) +  2\rho (1-D_1) \\
& \geq D_2 (1-2\rho) + D_2 - (1-\rho)^2\\
& \geq  (1-\rho)^2 \label{compfirst}
\end{align}
from which, along with \eqref{eqcompare} and \eqref{general1}, we find that  
\vspace{-0.15cm}\begin{equation}
R_{S_1 S_2} (D_1, D_2)- C_W(S_1, S_2) \leq 0 = R_{S_1|Z} (D_1) + R_{S_2|Z} (D_2). \label{complast}
\vspace{-0.15cm}\end{equation} 
Therefore, for all $\rho\leq 0.5$, Corollary~\ref{cor6} and the conditions \eqref{eq3LW1}-\eqref{eq3LW2} are equivalent. By comparing \eqref{eqcompare} with \eqref{general1}, we can show that, Corollary~\ref{cor6} is equivalent to \eqref{eq3LW1}-\eqref{eq3LW2} if
\begin{equation} 
\rho \in \left\{\rho: \Delta - \sqrt{\frac{D_1+D_2}{2} -1+ \Delta^2} \leq \rho \leq \Delta + \sqrt{\frac{D_1+D_2}{2} -1+ \Delta^2}, \;\frac{D_1+D_2}{2} + \Delta^2 \geq 1  \right\} \end{equation} 
where $\Delta \triangleq \frac{1+\sqrt{(1-D_1)(1-D_2)}}{2}$. 
We therefore find that the necessary conditions from Corollary~\ref{cor6} are at least as tight as conditions \eqref{eq3LW1}-\eqref{eq3LW2} in all regions but $\mathcal{E}$, $\mathcal{D}$, and $\mathcal{F}$. 

\begin{remark}We note that Corollary~\ref{cor6} is not necessarily strictly tighter in any of these regions, since the necessary conditions involve also the RHS of \eqref{eq3GGY1}-\eqref{eq3GGY2} and \eqref{eq3LW1}-\eqref{eq3LW2}, which  can be used to claim the impossibility of achieving certain distortion pairs based on the relative value of the rate distortion functions with respect to the rate region characterized by the RHS. It is possible that, even though the LHS of Corollary~\ref{cor6} is lower than the LHS of \eqref{eq3LW1}-\eqref{eq3LW2}, either both or none of the necessary conditions may be satisfied, leading exactly to the same conclusion regarding the achievability of the corresponding distortion pair. 
\end{remark}

\section{Proof of Proposition~\ref{propbounds2}}\label{appendixE}
Consider $D_1=0.3$, $\rho=0.5$, and $P=1$. For this case, \eqref{LT} holds with equality when $D_2 =  0.625$, and $(0.3,  0.625)\in \mathcal{B}$. Accordingly, no distortion pair $(0.3,  D_2)$, with $0.5\leq D_2 < 0.625$, satisfies \eqref{LT}. The necessary conditions of Corollary~\ref{cor6} for $(D_1, D_2)\in \mathcal{B}$ are given by
\begin{align} 
 \frac{1}{2}\log \left (\frac{1-\rho}{D_1}\right )  &\leq \frac{1}{2} \log ( 1 + \beta_1 P + \beta_2 P ) \label{eq3GGY123}\\
 \frac{1}{2}\log \left (\frac{1-\rho^2}{D_1D_2}\right)  &\leq \frac{1}{2} \log ( 1 + 2P + 2P \sqrt{(1-\beta_1)(1-\beta_2)} ). \label{eq3GGY223}
\end{align}
By defining $\hat{\alpha} \triangleq  \frac{1-\rho}{D_1}$, and setting $\beta_1 = \frac{\hat{\alpha}-1}{P}-\beta_2$, which satisfies \eqref{eq3GGY123}, condition \eqref{eq3GGY223} becomes,
\begin{equation}\label{pol323}
-\beta_2^2 + \frac{\hat{\alpha}-1}{P}\beta_2 + 1- \frac{\hat{\alpha}-1}{P} -\hat{\theta} \geq 0,
\end{equation}
where $\hat{\theta} \triangleq \left(\frac{1-\rho^2}{2PD_1D_2}  -\frac{1}{2P}-1\right)^2$. The LHS of \eqref{pol32} is concave, and attains its maximum value at $\beta_2 = \frac{\hat{\alpha}-1}{2P}=0.3333$. The corresponding $\beta_1$ is computed from $\beta_1 = \frac{\hat{\alpha}-1}{P}-\beta_2=0.3333$.  From \eqref{pol323}, it can be shown that Corollary~\ref{cor6} is satisfied whenever $D_2\geq 0.5769$. Accordingly, for the distortion pairs $(0.3,D_2)$ with $0.5769\leq D_2 < 0.625$, the necessary conditions of Corollary~\ref{cor6} are satisfied whereas the bound in \eqref{LT} is not.

\section{Proof of Proposition~\ref{propbound3}}\label{appendixF}

Let $D_1=0.25$, $\alpha=0.2$, $P=0.9$, and $0\leq D_2\leq \frac{\alpha}{2(1-\alpha)}$. 
Consider initially the condition from \eqref{LTDSBS}. Let $D_2 = 0.003$ and observe that for this case $R_{S_1 S_2} (D_1, D_2)=1-h(D_2)$. Then, 
\begin{equation}
R_{S_1 S_2}( D_1, D_2)  =0.9705 \leq \frac{1}{2} \log (1 + 2P(1+\rho_{max})) = 0.978,
\end{equation}
hence  \eqref{LTDSBS} is satisfied for all $D_2\geq 0.003$. Next, consider the conditions from \eqref{LWDS11}-\eqref{LWDS12}. Let $D_2 = 0.003$ and $\beta_1 = \frac{2^{2(2h(\theta)-h(D_2)-h(\alpha))}-1}{P} - \beta_2$ and observe that \eqref{LWDS11} is satisfied. By rearranging \eqref{LWDS11}-\eqref{LWDS12}, we obtain
\begin{equation}
-\beta_2^2 \!+\!\frac{\left(2^{2(2h(\theta)-h(D_2)-h(\alpha))}-1\right)}{P} \beta_2 + 1 \!-\! \frac{\left(2^{2(2h(\theta)-h(D_2)-h(\alpha))}\!-\!1\right)}{P} - \left(\frac{2^{2(1-h(D_2))}-1}{2P} \!-\! 1\right)^2\geq 0
\end{equation}
whose LHS reaches its maximum value $0.2344$ at $\beta_2 = \frac{2^{2(2h(\theta)-h(D_2)-h(\alpha))}-1}{2P} = 0.2462$. Therefore, necessary conditions  \eqref{LWDS11}-\eqref{LWDS12} are satisfied for all $D_2\geq 0.003$.

Next, consider the necessary conditions in \eqref{LWDS21}-\eqref{LWDS22}. Similar to the previous case, let $D_2 = 0.003$ and $\beta_1 = \frac{2^{2\left(\frac{\alpha}{1-\alpha}h(\alpha)-h(D_2)\right)}-1}{P} - \beta_2$ which satisfies \eqref{LWDS21}. Rearrange \eqref{LWDS21}-\eqref{LWDS22} to obtain
\begin{equation}
-\beta_2^2 +\frac{\left(2^{2\left(\frac{\alpha}{1-\alpha}h(\alpha)-h(D_2)\right)}-1\right)}{P} \beta_2 + 1 - \frac{\left(2^{2\left(\frac{\alpha}{1-\alpha}h(\alpha)-h(D_2)\right)}-1\right)}{P} - \left(\frac{2^{2(1-h(D_2))}-1}{2P} - 1\right)^2\geq 0
\end{equation}
whose LHS reaches a maximum of $0.4242$ at $\beta_2 = \frac{2^{2\left(\frac{\alpha}{1-\alpha}h(\alpha)-h(D_2)\right)}-1}{2P} = 0.1294$. Hence, necessary conditions from \eqref{LWDS21}-\eqref{LWDS22} are satisfied for all $D_2\geq 0.003$. 

Lastly, consider the necessary conditions from Corollary~\ref{cor6} and let $D_2=0.003$. From \eqref{eq3GGY1}, we have $\beta_1\geq \frac{2^{2(R_{S_1|Z} (D_1)  + R_{S_2|Z} (D_2))}-1}{P}-\beta_2$, from which, by combining with \eqref{eq3GGY2}, we obtain 
\begin{align}
-\beta_2^2 +\frac{\left(2^{2(R_{S_1|Z} (D_1)  + R_{S_2|Z} (D_2))}-1\right)}{P}\beta_2 + 1 - &\frac{\left(2^{2(R_{S_1|Z} (D_1)  + R_{S_2|Z} (D_2))}-1\right)}{P}   
 \notag \\
&
\qquad - \left(\frac{2^{2R_{S_1 S_2} (D_1, D_2)}-1}{2P} -1\right)^2\geq 0
\end{align}
and observe that the polynomial on the LHS attains its maximum value $-0.0247$ at $\beta_2 = \frac{\left(2^{2(R_{S_1|Z} (D_1)  + R_{S_2|Z} (D_2))}-1\right)}{2P} = 0.4442$. Hence, for this example,  Corollary~\ref{cor6} cannot be satisfied for any $0\leq \beta_1, \beta_2\leq 1$. We therefore conclude that there exist distortion pairs for which the two necessary conditions are satisfied while Corollary~\ref{cor6} is not.

\bibliographystyle{IEEEtran}
\bibliography{IEEEabrv,ref}

\begin{thebibliography}{10}
\providecommand{\url}[1]{#1}
\csname url@samestyle\endcsname
\providecommand{\newblock}{\relax}
\providecommand{\bibinfo}[2]{#2}
\providecommand{\BIBentrySTDinterwordspacing}{\spaceskip=0pt\relax}
\providecommand{\BIBentryALTinterwordstretchfactor}{4}
\providecommand{\BIBentryALTinterwordspacing}{\spaceskip=\fontdimen2\font plus
\BIBentryALTinterwordstretchfactor\fontdimen3\font minus
  \fontdimen4\font\relax}
\providecommand{\BIBforeignlanguage}[2]{{%
\expandafter\ifx\csname l@#1\endcsname\relax
\typeout{** WARNING: IEEEtran.bst: No hyphenation pattern has been}%
\typeout{** loaded for the language `#1'. Using the pattern for}%
\typeout{** the default language instead.}%
\else
\language=\csname l@#1\endcsname
\fi
#2}}
\providecommand{\BIBdecl}{\relax}
\BIBdecl

\bibitem{cover1980multiple}
T.~Cover, A.~E. Gamal, and M.~Salehi, ``Multiple access channels with
  arbitrarily correlated sources,'' \emph{IEEE Transactions on Information
  Theory}, vol.~26, no.~6, pp. 648--657, 1980.

\bibitem{Shannon}
C.~E. Shannon, ``A mathematical theory of communication,'' \emph{Bell System
  Technical Journal}, vol.~27, pp. 379--423, 625--656, July, Oct. 1948.

\bibitem{xiao2007multiterminal}
J.-J. Xiao and Z.-Q. Luo, ``Multiterminal source--channel communication over an
  orthogonal multiple-access channel,'' \emph{IEEE Transactions on Information
  Theory}, vol.~53, no.~9, pp. 3255--3264, 2007.

\bibitem{gunduz2007correlated}
D.~G{\"u}nd{\"u}z and E.~Erkip, ``Correlated sources over an asymmetric
  multiple access channel with one distortion criterion,'' \emph{41st Annual
  Conference on Information Sciences and Systems, CISS'07}, pp. 325--330,
  Baltimore, MD, 2007.

\bibitem{gunduz2009source}
D.~G{\"u}nd{\"u}z, E.~Erkip, A.~Goldsmith, and H.~V. Poor, ``Source and channel
  coding for correlated sources over multiuser channels,'' \emph{IEEE
  Transactions on Information Theory}, vol.~55, no.~9, pp. 3927--3944, Sep.
  2009.

\bibitem{minero2015unified}
P.~Minero, S.~H. Lim, and Y.-H. Kim, ``A unified approach to hybrid coding,''
  \emph{IEEE Transactions on Information Theory}, vol.~61, no.~4, pp.
  1509--1523, Apr. 2015.

\bibitem{lapidoth2010sending}
A.~Lapidoth and S.~Tinguely, ``Sending a bivariate {G}aussian over a {G}aussian
  {MAC},'' \emph{IEEE Transactions on Information Theory}, vol.~56, no.~6, pp.
  2714--2752, Jun. 2010.

\bibitem{jain2012energy}
A.~Jain, D.~Gunduz, S.~R. Kulkarni, H.~V. Poor, and S.~Verd{\'u},
  ``Energy-distortion tradeoffs in {Gaussian} joint source-channel coding
  problems,'' \emph{IEEE Transactions on Information Theory}, vol.~58, no.~5,
  pp. 3153--3168, May 2012.

\bibitem{kang2011new}
W.~Kang and S.~Ulukus, ``A new data processing inequality and its applications
  in distributed source and channel coding,'' \emph{IEEE Transactions on
  Information Theory}, vol.~57, no.~1, pp. 56--69, 2011.

\bibitem{7541654}
A.~Lapidoth and M.~Wigger, ``A necessary condition for the transmissibility of
  correlated sources over a {MAC},'' \emph{IEEE International Symposium on
  Information Theory, ISIT'16}, pp. 2024--2028, Barcelona, Spain, 2016.

\bibitem{yu2016distortion}
L.~Yu, H.~Li, and C.~W. Chen, ``Distortion bounds for transmitting correlated
  sources with common part over {MAC},'' in \emph{54th Annual Allerton
  Conference on Communication, Control, and Computing}, 2016, pp. 687--694.

\bibitem{tian2014optimality}
C.~Tian, J.~Chen, S.~N. Diggavi, and S.~Shamai, ``Optimality and approximate
  optimality of source-channel separation in networks,'' \emph{IEEE
  Transactions on Information Theory}, vol.~60, no.~2, pp. 904--918, Feb 2014.

\bibitem{tian2017matched}
C.~Tian, J.~Chen, S.~Diggavi, and S.~Shamai, ``Matched multiuser {Gaussian}
  source channel communications via uncoded schemes,'' \emph{IEEE Transactions
  on Information Theory}, vol.~63, no.~7, pp. 4155--4171, July 2017.

\bibitem{ozarow80}
L.~Ozarow, ``On a source-coding problem with two channels and three
  receivers,'' \emph{Bell Syst. Tech. Journal}, vol.~59, no.~10, pp.
  1909--1921, 1980.

\bibitem{wagnerIT08}
A.~B. Wagner and V.~Anantharam, ``An improved outer bound for multiterminal
  source coding,'' \emph{IEEE Transactions on Information Theory}, vol.~54,
  no.~5, pp. 1919--1937, 2008.

\bibitem{gray1972conditional}
R.~M. Gray, \emph{Conditional rate-distortion theory}.\hskip 1em plus 0.5em
  minus 0.4em\relax Inf. Sys. Laboratory, Stanford Electronics Laboratories,
  1972.

\bibitem{shamai1998systematic}
S.~Shamai, S.~Verd{\'u}, and R.~Zamir, ``Systematic lossy source/channel
  coding,'' \emph{IEEE Transactions on Information Theory}, vol.~44, no.~2, pp.
  564--579, Mar. 1998.

\bibitem{gacs1973common}
P.~G{\'a}cs and J.~K{\"o}rner, ``Common information is far less than mutual
  information,'' \emph{Problems of Control and Information Theory}, vol.~2,
  no.~2, pp. 149--162, 1973.

\bibitem{wyner1975common}
A.~D. Wyner, ``The common information of two dependent random variables,''
  \emph{IEEE Transactions on Information Theory}, vol.~21, no.~2, pp. 163--179,
  Mar. 1975.

\bibitem{wyner1976rate}
A.~D. Wyner and J.~Ziv, ``The rate-distortion function for source coding with
  side information at the decoder,'' \emph{IEEE Transactions on Information
  Theory}, vol.~22, no.~1, pp. 1--10, Apr. 1976.

\bibitem{bross2008gaussian}
S.~I. Bross, A.~Lapidoth, and M.~A. Wigger, ``The {G}aussian {MAC} with
  conferencing encoders,'' \emph{IEEE International Symposium on Information
  Theory, ISIT'08}, pp. 2702--2706, Toronto, Canada, 2008.

\bibitem{xu2016lossy}
G.~Xu, W.~Liu, and B.~Chen, ``A lossy source coding interpretation of {Wyner's}
  common information,'' \emph{IEEE Transactions on Information Theory},
  vol.~62, no.~2, pp. 754--768, Feb. 2016.

\bibitem{wyner1978rate}
A.~D. Wyner, ``The rate-distortion function for source coding with side
  information at the decoder{-II}: General sources,'' \emph{Information and
  Control}, vol.~38, no.~1, pp. 60--80, Jul. 1978.

\bibitem{xiao2005compression}
J.-J. Xiao and Z.-Q. Luo, ``Compression of correlated {G}aussian sources under
  individual distortion criteria,'' \emph{43rd Annual Allerton Conference on
  Communication, Control, and Computing}, Monticello, IL, 2005.

\bibitem{steinberg2009coding}
Y.~Steinberg, ``Coding and common reconstruction,'' \emph{IEEE Transactions on
  Information Theory}, vol.~55, no.~11, pp. 4995--5010, Nov. 2009.

\bibitem{nayak2010successive}
J.~Nayak, E.~Tuncel, D.~G{\"u}nd{\"u}z, and E.~Erkip, ``Successive refinement
  of vector sources under individual distortion criteria,'' \emph{IEEE
  Transactions on Information Theory}, vol.~56, no.~4, pp. 1769--1781, 2010.

\bibitem{anantharam2014hypercontractivity}
V.~Anantharam, A.~A. Gohari, S.~Kamath, and C.~Nair, ``On hypercontractivity
  and a data processing inequality,'' \emph{IEEE International Symposium on
  Information Theory, ISIT'14}, pp. 3022--3026, Honolulu, HI, 2014.

\bibitem{wagner2011distributed}
A.~B. Wagner, B.~G. Kelly, and Y.~Altu{\u{g}}, ``Distributed rate-distortion
  with common components,'' \emph{IEEE Transactions on Information Theory},
  vol.~57, no.~7, pp. 4035--4057, Jul. 2011.

\bibitem{gastpar2004wyner}
M.~Gastpar, ``The {W}yner-{Z}iv problem with multiple sources,'' \emph{IEEE
  Transactions on Information Theory}, vol.~50, no.~11, pp. 2762--2768, Nov.
  2004.

\bibitem{slepian1973coding}
D.~Slepian and J.~K. Wolf, ``A coding theorem for multiple access channels with
  correlated sources,'' \emph{Bell System Technical Journal}, vol.~52, no.~7,
  pp. 1037--1076, Sep. 1973.

\bibitem{cover2012elements}
T.~M. Cover and J.~A. Thomas, \emph{Elements of information theory}.\hskip 1em
  plus 0.5em minus 0.4em\relax John Wiley \& Sons, 2012.

\end{thebibliography}

\end{document}